\documentclass[11pt]{article}
\usepackage[utf8]{inputenc}
\usepackage[english]{babel}
\usepackage[margin=1.15in]{geometry}
\usepackage{parskip}
\usepackage{amsmath}
\usepackage{amsthm}
\usepackage{amsfonts}
\usepackage{graphicx}
\usepackage{float}
\usepackage{pgfplots}
\pgfplotsset{compat=1.14, set layers}
\usepackage{bm}
\usepackage{mathtools}
\usepackage{subfig}
\usepackage{enumitem}

\usepackage{lipsum}
\usepackage{setspace}
\usepackage{caption}

\usepackage{float}
\usepackage{mathtools}
\usepackage{todonotes}
\usepackage{parskip}
\usepackage{xspace}
\usepackage[bottom]{footmisc}

\definecolor{refcolor}{rgb}{0.23, 0.27, 0.29}
\usepackage[breaklinks,colorlinks,urlcolor={refcolor},citecolor={refcolor},linkcolor={refcolor}]{hyperref}
\usepackage[capitalize, nameinlink]{cleveref}

\newtheorem{theorem}{Theorem}[section]
\newtheorem{lemma}[theorem]{Lemma}

\newtheorem{corollary}[theorem]{Corollary}
\newtheorem{definition}[theorem]{Definition}

\newtheorem{observation}[theorem]{Observation}

\newcommand{\larr}{\xleftarrow{}}
\newcommand{\fA}{\mathcal A}
\newcommand{\fC}{\mathcal C}
\newcommand{\fG}{\mathcal G}
\newcommand{\fP}{\mathcal P}
\newcommand{\ps}{\mathcal P}
\newcommand{\undec}{\texttt{undecided}}
\newcommand{\incid}{\operatorname{inc}}
\newcommand{\aps}{\operatorname{act}}
\newcommand{\pairs}{\operatorname{pairs}}
\newcommand{\pred}{\operatorname{pred}}
\newcommand{\first}{\operatorname{first}}
\newcommand{\last}{\operatorname{last}}
\newcommand{\comp}{M}
\newcommand{\merge}{\operatorname{merge}}
\newcommand{\prev}{\operatorname{prev}}
\newcommand{\new}{\operatorname{new}}

\newcommand{\stays}{\operatorname{stay}}
\newcommand{\timey}{\operatorname{time}}
\newcommand{\fin}{\operatorname{fin}}

\newcommand{\betw}{\operatorname{between}}
\newcommand{\done}{\operatorname{done}}
\newcommand{\succc}{\operatorname{succ}}
\newcommand{\atwo}{\mathcal A'_{\operatorname{II}}}

\newcommand{\inn}{{\operatorname{in}}}
\newcommand{\out}{{\operatorname{out}}}
\newcommand{\sinn}{\Sigma_{\inn}}
\newcommand{\sout}{\Sigma_{\out}}
\newcommand{\phinn}{\phi_{\inn}}
\newcommand{\phout}{\phi_{\out}}
\newcommand{\ginn}{g_{\inn}}
\newcommand{\gout}{g_{\out}}
\newcommand{\pinn}{p_{\inn}}
\newcommand{\pout}{p_{\out}}
\newcommand{\noco}{\mathcal N}
\newcommand{\edco}{\mathcal E}
\newcommand{\gee}{g}
\newcommand{\rooot}{r}

\newcommand{\mpc}{\textsf{MPC}\xspace}

\newcommand{\lcl}{\textsf{LCL}\xspace}
\newcommand{\lcls}{\textsf{LCL}s\xspace}
\newcommand{\local}{\textsf{LOCAL}\xspace}

\newcommand{\congest}{\textsf{CONGEST}\xspace}

\newcommand{\class}{\textsf{CLASS\xspace}}

\begin{document}

\begin{center}
    {\LARGE \bf Towards a Complexity Classification of \lcl Problems in Massively Parallel Computation} \\
\vspace{1cm}

\begin{minipage}[H]{5cm}
\begin{center}
    {\large \bf Sebastian Brandt} \\
    {\large CISPA Helmholtz Center for Information Security} \\ \href{mailto:brandt@cispa.de}{\texttt{brandt@cispa.de}}
\end{center}
\end{minipage}
\begin{minipage}[H]{5cm}
\begin{center}
    {\large \bf Rustam Latypov\footnotemark} \\
    {\large Aalto University} \\ \href{mailto:rustam.latypov@aalto.fi}{\texttt{rustam.latypov@aalto.fi}}
    \vspace{5.2mm}
\end{center}
\end{minipage}
\begin{minipage}[H]{5cm}
\begin{center}
    {\large \bf Jara Uitto} \\
    {\large Aalto University} \\
    \href{mailto:jara.uitto@aalto.fi}{\texttt{jara.uitto@aalto.fi}}
    \vspace{5.2mm}
\end{center}
\end{minipage}

\vspace{4mm}
\begin{center}
    \today \\ 
\end{center}

\vspace{1cm}
\begin{minipage}[H]{13.3cm}
\begin{center}
    {\bf Abstract} \\ \vspace{4mm}
\end{center}

In this work, we develop the low-space Massively Parallel Computation (\mpc) complexity landscape for a family of fundamental graph problems on trees. We present a general method that solves most locally checkable labeling (\lcl) problems exponentially faster in the low-space \mpc model than in the \local message passing model. In particular, we show that all solvable \lcl problems on trees can be solved in $O(\log n)$ time (high-complexity regime) and that all \lcl problems on trees with deterministic complexity $n^{o(1)}$ in the \local model can be solved in $O(\log \log n)$ time (mid-complexity regime). 
We observe that obtaining a greater speed-up than from $n^{o(1)}$ to $\Theta(\log \log n)$ is conditionally impossible, since the problem of 3-coloring trees, which is a \lcl problem with \local time complexity $n^{o(1)}$, has a conditional \mpc lower bound of $\Omega(\log \log n)$ [Linial, FOCS'87; Ghaffari, Kuhn and Uitto, FOCS'19].
We emphasize that we solve \lcl problems on \textit{constant}-degree trees, and that our algorithms are deterministic, component-stable, and work in the low-space \mpc model, where local memory is $O(n^\delta)$ for $\delta \in (0,1)$ and global memory is $O(m)$. \\

For the high-complexity regime, there are two key ingredients. One is a novel $O(\log n)$-time tree rooting algorithm, which may be of independent interest. The other is a novel pointer-chain technique and analysis that allows us to solve any solvable \lcl problem on trees in $O(\log n)$ time. For the mid-complexity regime, we adapt the approach by Chang and Pettie [FOCS'17], who gave a canonical \local algorithm for solving \lcl problems on trees.

\end{minipage}
\end{center}

\vfill
\thispagestyle{empty}
\footnotetext{Supported in part by the Academy of Finland, Grant 334238}

\clearpage
\setcounter{page}{1}

\section{Introduction}

The Massively Parallel Computation (\mpc) model, introduced in~\cite{KarloffSV10} and later refined by~\cite{mpcrefine1, mpcrefine2, broadcast}, is a mathematical abstraction of modern data processing platforms such as MapReduce~\cite{dg04}, Hadoop~\cite{White:2012}, Spark~\cite{ZahariaCFSS10}, and Dryad~\cite{Isard:2007}.
Recently, tremendous progress has been made on fundamental \emph{graph problems} in this model, such as maximal independent set (MIS), maximal matching (MM), and coloring problems.
All these problems, and many others, fall under the umbrella of \emph{Locally Checkable Labeling (\lcl)} problems, which are defined through a set of feasible configurations from the viewpoint of each individual node.
These problems serve as abstractions for fundamental primitives in large-scale graph processing and have recently gained a lot of attention~\cite{lclcongest, Balliu21, detcol,spanner, balliu2020, Chang2020,GGC20}.

The goal of our work is to go beyond individual problems and initiate a \emph{systematic study} of central graph problems in the \mpc model.
The holy grail is to \emph{automate} the design of algorithms in the following sense:
Given a problem $\Pi$, we automatically obtain an \mpc algorithm that solves the problem efficiently.
In the distributed \local model of computing, such a result is known in the case of \emph{rooted trees}~\cite{Balliu21} and the techniques automatically yield provably \emph{optimal} algorithms for the class of \lcl problems.

In our work, we focus on \lcl problems in (unrooted) trees in the most restricted \emph{low-space} \mpc model with \emph{linear} total memory that is the most \emph{scalable} variant of the \mpc model.
Our results provide an automatic method that, for a wide range of \lcl problems, yields an algorithm that solves the given problem \emph{exponentially} faster than an optimal distributed counterpart.
The resulting algorithms are \emph{component-stable}, which implies that the solutions in individual connected components are independent of the other components.

\subsection{The \mpc Model}

In the \mpc model, we have $M$ machines who communicate in an all-to-all fashion. We focus on problems where the input is modeled as a graph with $n$ vertices, $m$ edges and maximum degree $\Delta$; we call this graph the \textit{input graph}. Each node has a unique ID of size $O(\log n)$ bits from a domain $\{1,2,\dots,N\}$, where $N = \text{poly}(n)$. Each node and its incident edges are hosted on a machine with $O(n^\delta)$ \textit{local memory} capacity, where $\delta \in (0,1)$ and the units of memory are words of $O(\log n)$ bits. When the local memory is bounded by $O(n^\delta)$, the model is called \textit{low-space} (or \textit{sublinear}). The number of machines is chosen such that $M = \Theta(m/n^\delta)$. For trees, where $m=\Theta(n)$, this results in $\Theta(n^{1-\delta})$ machines. For simplicity, we assume that each machine $i$ simulates one virtual machine for each node and its incident edges that $i$ hosts, such that the local memory restriction becomes that no virtual machine can use more than $O(n^\delta)$ memory. In our work, we only consider trees with constant degree, and hence we can assume that the edges incident to each node fit into the memory of the hosting virtual machine.

During the execution of an \mpc algorithm, computation is performed in synchronous, fault-tolerant rounds. In each round, every machine performs some (unbounded) computation on the locally stored data, then sends/receives messages to/from any other machine in the network. Each message is sent to exactly one other machine specified by the sending machine. All messages sent and received by each machine in each round, as well as the output, have to fit into local memory. The time complexity is the number of rounds it takes to solve a problem. Upon termination, each node (resp.\ its hosting machine) must know its own part of the solution.

The \textit{global memory} (or \textit{total memory}) is the sum over the local memory over all machines. In this work, we restrict the global memory to linear in the number of edges, i.e., $O(m)$, which is the strictest possible as it is only enough to store a constant number of copies of the input graph. Note that if we were to allow superlinear $O(m^{1+\delta})$ global memory in constant-degree trees, many \local algorithms with complexity $O(\log n)$ could be exponentially sped up in the low-space \mpc model by applying the well-known graph exponentiation technique by Lenzen and Wattenhofer~\cite{wattenhofer}. A crucial challenge that comes with the linear global memory restriction is that only a small fraction of $n^{1 - \delta}$ of the (virtual) machines can simultaneously utilize all of their available local memory.

\subsection{The Complexity Landscapes of \local and \congest}

In the last decade, there has been tremendous progress in understanding the complexities of locally checkable problems in various models of distributed and parallel computing. A prime example is the \local model~\cite{linial}, where the input graph corresponds to a message passing system, and the nodes must output their part of the solution according only to local information about the graph. Another example is the \congest model, which is a \local model variant where the message size is restricted to $O(\log n)$~\cite{peleg}. In these two models, the \textit{whole} complexity landscape of \lcl problems is now understood for some important graph families. For instance, a rich line of work~\cite{NaorS95, tree1, chang, tree2, tree3, balliu2020, Chang2020} recently came to an end when a complexity gap between $\omega(1)$ and $o(\log^* n)$ was proved~\cite{lclcomplete}, completing the randomized/deterministic complexity landscape of \lcl problems in the \local model for trees. In the randomized/deterministic \local and \congest models, recent work showed that the complexity landscapes of \lcl problems for rooted regular trees are fully understood~\cite{Balliu21}, while the complexity landscapes of \lcl problems in the \local model for rings and tori have already been know for some while~\cite{Brandt17}. Even for general (constant-degree) graphs, the LOCAL complexity landscape of \lcl problems is almost fully understood~\cite{NaorS95, sinkless16, CKP19, chang, GhaSu17, FischerG17, GHK18, BHKLOS18, tree3, RozhonG20, GhaffariGR21}, only missing a small part of the picture related to the randomized complexity of Lovász Local Lemma (LLL).

As the most relevant results related to our work, we want to highlight two papers. In the \congest model, on trees, the complexity of an \lcl problem is asymptotically equal to its complexity in the \local model whereas the same does not hold in general graphs, as shown in~\cite{lclcongest}. It is worth noting that in order to prove the asymptotic equality between \local and \congest, the authors of~\cite{lclcongest} use very similar, but independently developed, techniques to ours in the mid-complexity regime. Furthermore, we want to accentuate the relevance of the work of Chang and Pettie~\cite{chang}, as we simulate their algorithms in our mid-complexity regime.

\subsection{Our Contributions}

\begin{figure}
	\centering
	\includegraphics[width=\textwidth]{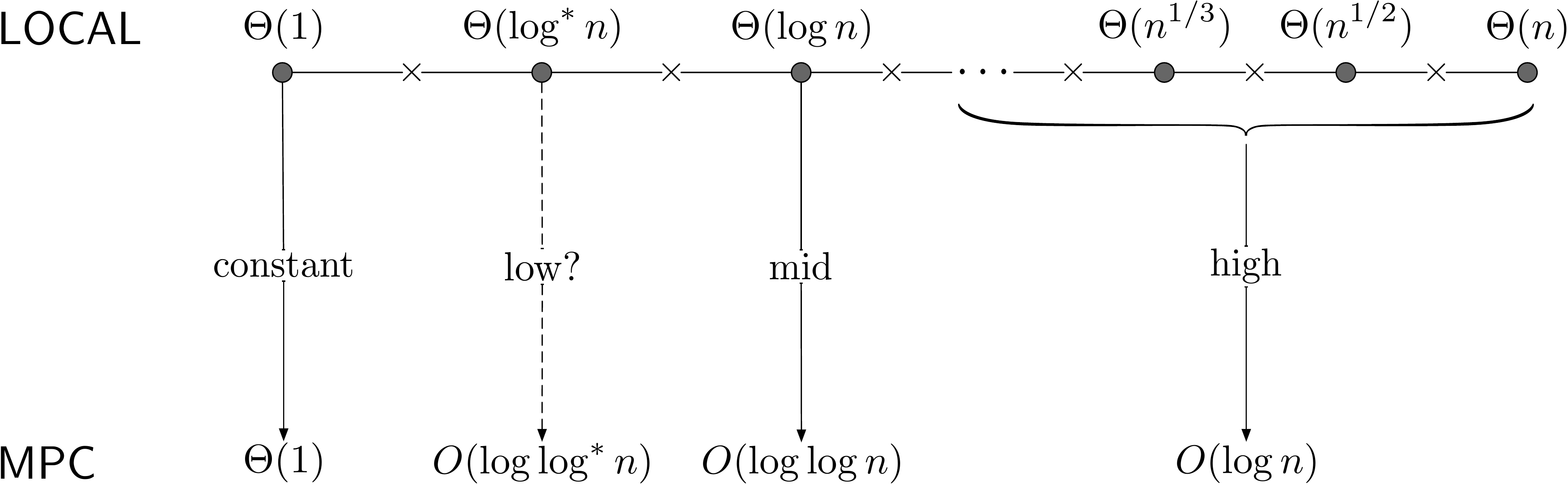}
	\caption{\textbf{Top:} The complete complexity landscape of \lcl problems for constant-degree trees in the deterministic \local model. \textbf{Bottom:} The complexity landscape of \lcl problems for constant-degree trees in the deterministic low-space \mpc model, comprised of the constant-complexity, low-complexity, mid-complexity and the high-complexity regimes.} 
	\label{fig:map}
\end{figure}

Our main contribution is showing that almost all \lcl problems (see \Cref{def:lcl}) on trees can be solved deterministically and exponentially faster in the low-space \mpc model than in the \local model. In particular, we prove the following.

\vspace{1mm}
\fcolorbox{black}{black!0.5}{\begin{minipage}{0.985\textwidth}
\textbf{High-complexity regime.} All \lcl problems on trees can be solved deterministically in $O(\log n)$ time in the low-space \mpc model using linear global memory. (\Cref{thm:high})
\end{minipage}}

\vspace{3mm}

\fcolorbox{black}{black!0.5}{\begin{minipage}{0.985\textwidth}
\textbf{Mid-complexity regime.} All \lcl problems on trees with deterministic time complexity $n^{o(1)}$ in the \local model can be solved deterministically in $O(\log \log n)$ time in the low-space \mpc model using linear global memory. (\Cref{thm:mid})
\end{minipage}}
\vspace{1mm}

A sensible way to interpret our results from the point of view of the \local model is as follows.
By the aforementioned rich line of work, we know that the complexity landscape of \lcl problems in the deterministic \local model for trees is comprised of the discrete complexity classes $\Theta(1)$, $\Theta(\log^*n)$, $\Theta(\log n)$, and $\Theta(n^{1/k})$ for any positive integer $k$.
Since all $\lcl$ problems fall into one of the aforementioned complexity classes in \local, it is natural to try to map said classes to exponentially faster \mpc classes, which is exactly how one can interpret our results. \Cref{fig:map} illustrates how our \mpc complexity regimes relate to the \local complexity classes.
We note that (i) the constant-complexity regime is trivial as we solve \lcl problems in \textit{constant}-degree trees and (ii) the low-complexity regime has been resolved in unpublished follow-up work by different authors~\cite{low-complex}.

We emphasize that all of our results use \emph{optimal} \mpc parameters in the sense that we work in the low-space setting with $O(n^\delta)$ words of local memory and $O(m)$ words of global memory. 



For 3-coloring trees, the classic work of Linial~\cite[Theorem 3.1]{linial} provides a \local lower bound of $\Omega(\log n)$, even when allowing shared randomness. Due to newer results~\cite{focs,componentstable}, this translates into a conditional $\Omega(\log \log n)$ lower bound for component-stable, low-space \mpc algorithms. This hardness result is conditioned on a widely believed conjecture in \mpc about the complexity of the connectivity problem, which asks to detect the connected components of a graph. 
Combined with the fact that $3$-coloring trees is clearly an $n^{o(1)}$-round solvable problem in \local, and our mid-complexity regime yields an $O(\log \log n)$ time complexity, we obtain a conditionally tight complexity of $\Theta(\log \log n)$.

It is worth noting that recent work~\cite{componentstable} showed a separation between component-stable and component-unstable algorithms, and that some (non-\lcl) problems can be solved much faster with unstable algorithms than with stable ones. 
While the results in~\cite{componentstable} cast uncertainty on the (unconditional) optimality of our framework, it is not clear whether the separation result holds for any \lcl problem.

\paragraph{Follow-up Work.}

In a follow-up work by a different set of authors, the low-complexity regime (\Cref{fig:map}) has been resolved.
The authors gave an $O(\log \log^* n)$ time low-space \mpc algorithm with linear global memory~\cite{low-complex}. Combining their result with our work, we obtain that all \lcl problems on trees can be solved (at least) exponentially faster in the \mpc model than in the \local model, crystallized in the following corollary.
 
\begin{corollary} \label{cor:all}
	Consider an \lcl problem on trees with deterministic time complexity $f(n)$ in the \local model. This problem can be solved deterministically in $O(\log f(n))$ time in the low-space \mpc model using linear global memory. Furthermore, assuming the connectivity conjecture, this bound is tight for component-stable algorithms.
\end{corollary}

\subsection{Key Techniques and Roadmap}
After concluding the current section with some further related work, we start off with some preliminaries in \Cref{sec:prelim}.
In \Cref{sec:rooting}, we give a novel \mpc algorithm for rooting a tree deterministically in $O(\log n)$ time, which we use as a subroutine in the high-complexity regime. In order to make due with linear global memory, we develop a \textit{path exponentiation} technique, in which nodes perform graph exponentiation on paths while only keeping at most two extra IDs in memory. We note that the rooting algorithm may be of independent interest. It is component-stable, and compatible with arbitrary degrees and forests (\Cref{lemma:rooting}).

In \Cref{sec:high}, we provide a constructive proof for the high-complexity regime: we explicitly provide, for any solvable \lcl, an algorithm $\fA$ that has a runtime of $O(\log n)$. On a high level, algorithm $\fA$ proceeds in $3$ phases. The first phase consists of rooting the input tree by using the aforementioned algorithm as a blackbox. In the second phase, roughly speaking, the goal is to compute, for a substantial number of nodes $v$, the set of output labels that can be output at $v$ such that the label choice can be extended to a (locally) correct solution in the subtree hanging from $v$. This is done in an iterative manner, proceeding from the leaves towards the root. The last phase consists in using the computed information to solve the given \lcl from the root downwards.

While this outline sounds simple, there are a number of intricate challenges that require the development of novel techniques, both in the design of the algorithm and its analysis. For instance, the depth of the input tree can be much larger than $\Theta(\log n)$ (which prevents us from performing the above ideas in a sequential manner),
and the storage of the required completability information in a standard implementation exceeds the available memory when using graph exponentiation.
Two of our key technical contributions are 1) the design of a process that allows for interleaving graph exponentiation steps (where, during the execution of a graph exponentiation process on some part of the input tree, different parts of the tree can join in the process, or even many graph exponentiation processes executed on individual parts of the tree can be merged, simultaneously or at different times, into one process during the execution) and 2) the design of a fine-tuned potential function for the analysis of the complex algorithm resulting from addressing the aforementioned issues and the highly non-sequential behavior arising from interleaving graph exponentiation steps.
We provide a more detailed overview of the challenges and our solutions in \Cref{sec:overviewhcr}.

In \Cref{sec:rc}, we give an $O(\log \log n)$ time \mpc algorithm for the rake-and-compress decomposition, which is based on simulating the \local decomposition algorithm by Chang and Pettie~\cite{chang}. Using the decomposition, we provide an algorithm for the mid-complexity regime in \Cref{sec:mid}, which relies heavily on the \local framework for solving \lcls by Chang and Pettie~\cite{chang}. Essentially, our algorithm consist of simulating their \local algorithm while paying special attention to the memory usage. Based on our rake-and-compress decomposition, we first pre-process the graph by dividing the nodes into \emph{batches} according to which layer (or partition) they belong to such that removing batch $i$ from the decomposition reduces the number of nodes by a factor of $\Delta^{2^i}$ ($\Delta$ denoting the maximum degree of the graph). Then, we alternate between simulating the \local algorithm and performing graph exponentiation in order to achieve the desired runtime and low memory.



Note that in trees, the number $n$ of nodes and the number $m$ of edges are asymptotically equal, and we may use them interchangeably throughout the paper when reasoning about global memory.

\subsection{Further Related Work}

For many of the classic graph problems, simple $O(\log n)$ time \mpc algorithms follow from classic literature in the \local model and \textsf{PRAM}~\cite{Alon86, linial, Luby85}. In particular in the case of bounded degree graphs, it is often straightforward to simulate algorithms from other models. However, it is usually desirable to get algorithms that run \emph{much faster} than their \local counterparts.
If the \mpc algorithms are given \emph{linear} $\Theta(n)$ or even \emph{superlinear} $\Theta(n^{1+\delta})$ local memory, fast algorithms are known for many classic graph problems.

In the sublinear (or low-space) model, \cite{Chang2019} provided a randomized algorithm for the $(\Delta + 1)$-coloring problem that, combined with the new network decomposition results~\cite{RozhonG20,GhaffariGR21}, yields an $O(\log \log \log n)$ \mpc algorithm, that is exponentially faster than its \local counterpart. A recent result by Czumaj, Davies, and Parter~\cite{detcol} provides a deterministic $O(\log \log \log n)$ time algorithm for the same problem using derandomization techniques. 
For many other problems, the current state of the art in the sublinear model is still far from the aforementioned exponential improvements over the \local counterparts, at least in the case of general graphs. For example, the best known MIS, maximal matching, $(1+\epsilon)$-approximation of maximum matching, and 2-approximation of minimum vertex cover algorithms run in $\widetilde{O}(\sqrt{\log \Delta} + \sqrt{\log \log n})$ time~\cite{GU19}, whereas the best known \local algorithm has a logarithmic dependency on $\Delta$~\cite{Ghaffari16}. For restricted graph classes, such as trees and graphs with small arboricity\footnote{The arboricity of a graph is the minimum number of disjoint forests into which the edges of the graph can be partitioned.} $\alpha$, better algorithms are known~\cite{sirocco, Behnezhad19}. Through a recent work by Ghaffari, Grunau and Jin, the current state of the art for MIS and maximal matching are $O(\sqrt{\log \alpha} \cdot \log \log \alpha + \log \log n)$ time algorithms using $\widetilde{O}(n + m)$ words of global memory~\cite{GGC20}. 


As for lower bounds, \cite{focs} gave conditional lower bounds of $\Omega(\log \log n)$ for component-stable sublinear \mpc algorithms for constant approximation of maximum matching and minimum vertex cover, and MIS. In addition, the authors provided a lower bound of $\Omega(\log \log \log n)$ for LLL. Their hardness results are conditioned on a widely believed conjecture in \mpc about the complexity of the connectivity problem, which asks to detect the connected components of a graph. It is argued that disproving this conjecture would imply rather strong and surprising implications in circuit complexity~\cite{Roughgarden18}. When assuming component stability, they also argue that all known algorithms in the literature are component-stable or can easily be made component-stable with no asymptotic increase in the round complexity. However, recent work~\cite{componentstable} gave a separation between stable and unstable algorithms, and that some particular problems (e.g., computing an independent set of size $\Omega(n/\Delta)$) can be solved faster with unstable algorithms than with stable ones.

It is also worth discussing the complexity of rooting a tree, as it is an important subroutine in our high-complexity regime. On the randomized side, \cite{sirocco} gave an $O(\log d \cdot \log \log n)$ time algorithm, where $d$ is the diameter of the graph. On the deterministic side, Coy and Czumaj~\cite{coy2021deterministic}, in a concurrent work with this paper, gave an $O(\log n)$ time algorithm using derandomization methods, which is the current state of the art. In \Cref{sec:rooting} we will provide a different, but also deterministic, $O(\log n)$ time rooting algorithm. We note that \cite{pathexp} uses similar techniques, but in a more general setting and in $\omega(\log n)$ time.

\section{Preliminaries} \label{sec:prelim}

We work with undirected, finite, simple graphs $G = (V,E)$ with $n=|V|$ nodes and $m=|E|$ edges such that $E \subseteq [V]^2$ and $V \cap E = \emptyset$. Let $\deg_G(v)$ denote the degree of a node $v$ in $G$ and let $\Delta$ denote the maximum degree of $G$. The distance $d_G(v,u)$ between two vertices $v,u$ in $G$ is the length of a shortest $v - u$ path in $G$; if no such path exists, we set $d_G(v, u) \coloneqq \infty$. The greatest distance between any two vertices in $G$ is the diameter of $G$, denoted by $\text{diam}(G)$. For a subset $S \subseteq V$, we use $G[S]$ to denote the subgraph of $G$ induced by nodes in $S$. Let $G^k$, where $k \in \mathbb{N}$, denote the $k$:th power of a graph $G$, which is another graph on the same vertex set, but in which two vertices are adjacent if their distance in $G$ is at most $k$. In the context of \mpc, $G^k$ is the resulting virtual graph after performing $\log k$ steps of graph exponentiation~\cite{wattenhofer}.

For each node $v$ and for every radius $k \in \mathbb{N}$, we denote the $k$-hop (or $k$-radius) neighborhood of $v$ as $N^k(v) = \{ u \in V : d(v,u) \leq k\}$. The topology of a neighborhood $N^k(v)$ of $v$ is simply $G[N^k(v)]$. However, with slight abuse of notation, we sometimes refer to $N^k(v)$ both as the node set and the subgraph induced by node set $N^k(v)$. Neighborhood topology knowledge is often referred to as vision, e.g., node $v$ sees $N^k(v)$.

\subsection{\lcl Definitions}

In their seminal work~\cite{NaorS95}, Naor and Stockmeyer introduced the notion of a locally checkable labeling problem (\lcl problem or just \lcl for short). The definition they provide restricts attention to problems where nodes are labeled (such as vertex coloring problems), but they remark that a similar definition can be given for problems where edges are labeled (such as edge coloring problems). A modern way to define \lcl problems that captures both of the above types of problems (and combinations thereof) labels \emph{half-edges} instead, i.e., pairs $(v,e)$ where $e$ is an edge incident to vertex $v$. Let us first define a half-edge labeling formally, and then provide this modern \lcl problem definition.

\begin{definition}[Half-edge labeling]\label{def:halfedge}
    A \emph{half-edge} in a graph $G = (V,E)$ is a pair $(v,e)$, where $v \in V$ is a vertex, and $e \in E$ is an edge incident to $v$.
    A half-edge $(v,e)$ is incident to some vertex $w$ if $v = w$.
    We denote the set of half-edges of $G$ by $H = H(G)$.
    A \emph{half-edge labeling} of $G$ with labels from a set $\Sigma$ is a function $g \colon H(G) \to \Sigma$.
\end{definition}

We distinguish between two kinds of half-edge labelings: \emph{input labelings} that are part of the input and \emph{output labelings} that are provided by an algorithm executed on input-labeled instances. Throughout the paper, we will assume that any considered input graph $G$ comes with an input labeling $\ginn \colon H(G) \to \sinn$ and will refer to $\sinn$ as the \emph{set of input labels}; if the considered \lcl problem does not have input labels, we can simply assume that $\sinn = \{\bot\}$ and that each node is labeled with $\bot$. Then, \Cref{def:solve} details how a correct solution for an \lcl problem is formally specified.

\begin{definition}[\lcl] \label{def:lcl}
    An \lcl problem, \lcl for short, is a quadruple $\Pi = (\sinn, \sout, r, \fP)$ where $\sinn$ and $\sout$ are finite sets (of input and output labels, respectively), $r \geq 1$ is an integer, and $\fP$ is a finite set of labeled graphs $(P, \pinn, \pout)$. The input and output labeling of $P$ are specified by $\pinn \colon H(P) \to \sinn$ and $\pout \colon H(P) \to \sout$, respectively.\footnote{Note that the original definition given in~\cite{NaorS95} considers \emph{centered graphs}; however, since we only consider trees, considering uncentered graphs instead suffices.}
\end{definition}


\begin{definition}[Solving an \lcl] \label{def:solve}
    Recall that $N^r(v)$ denotes the subgraph of $G$ induced by all nodes at distance at most $r$ from $v$. This naturally extends to labeled graphs. A \emph{correct solution} for an \lcl problem $\Pi = (\sinn, \sout, r, \fP)$ on a graph $(G, \ginn)$ labeled with elements from $\sinn$ is a half-edge labeling $\gout \colon H(G) \to \sout$ s.t.~for each node $v \in V(G)$, the neighborhood $N^r(v)$ in $(G, \ginn, \gout)$ is isomorphic to some member of $\fP$. We require that the isomorphism respects the input and output labelings of $N^r(v)$ and the member of $\fP$. We say that an algorithm $\fA$ \emph{solves} an \lcl problem $\Pi$ on a graph class $\fG$ if it provides a correct solution for $\Pi$ for every $G \in \fG$.
\end{definition}

Note that the \lcl definitions above require that graph class $\fG$ has constant degree.

It is often useful to rephrase a given \lcl in a way that minimizes the integer $r$ in the \lcl definition. In fact, since we only consider trees, any \lcl can be rephrased in a special form, called \emph{node-edge-checkable \lcl}, where $r$ is essentially set to $1$.\footnote{Arguably, this can be seen as $r = 1/2$, which might provide a better intuition.} While the formal definition of a node-edge-checkable \lcl appears complicated, the intuition behind it is simple: essentially, we have a list of allowed output label combinations around nodes, a list of allowed output label combinations on edges, and a list of allowed input-output label combinations, all of which a correct solution for the \lcl has to satisfy.

\begin{definition}[Node-edge-checkable \lcl]\label{def:nodeedge}
    Let $\Delta$ be some non-negative integer constant. A \emph{node-edge-checkable \lcl} is a quintuple $\Pi = (\sinn, \sout, \noco, \edco, \gee)$ where $\sinn$ and $\sout$ are finite sets, $\noco = \{\noco_1, \dots, \noco_{\Delta} \}$ consists of sets $\noco_i$ of cardinality-$i$ multisets with elements from $\sout$, $\edco$ is a set of cardinality-$2$ multisets with elements from $\sout$, and $\gee \colon \sinn \to 2^{\sout}$ is a function mapping input labels to sets of output labels.
    We call $\noco_1 \cup \dots \cup \noco_{\Delta}$ and $\edco$ the \emph{node constraint} and \emph{edge constraint} of $\Pi$, respectively.
    Furthermore, we call each element of $\noco$ a \emph{node configuration}, and each element of $\edco$ an \emph{edge configuration}.
    
    For a node $v$, denote the half-edges of the form $(v,e)$ for some edge $e$ by $h_1^v, \dots, h_{\deg(v)}^v$ (in arbitrary order).
    For an edge $e$, denote the half-edges of the form $(v,e)$ for some node $v$ by $h_1^e, h_2^e$ (in arbitrary order).
    A correct solution for $\Pi$ is a half-edge labeling $\gout \colon H(G) \to \sout$ such that
    \begin{enumerate}
        \item for each node $v$, the multiset of outputs assigned by $\gout$ to $h_1^v, \dots, h_{\deg(v)}^v$ is an element of $\noco_{\deg(v)}$,
        \item for each edge $e$, the cardinality-$2$ multiset of outputs assigned by $\gout$ to $h_1^e, h_2^e$ is an element of $\edco$, and
        \item for each half-edge $h \in H(G)$, we have $\gout(h) \in \gee(\iota)$, where $\iota = \ginn(h)$ is the input label assigned to $h$.
    \end{enumerate}
\end{definition}

On trees, each \lcl $\Pi$ (with parameter $r$ in its definition) can be transformed into a node-edge-checkable \lcl $\Pi'$ by the standard technique of requiring each node $v$ to output, on each incident half-edge $h$, an encoding of its entire $r$-hop neighborhood (including input labels, output labels, and a marker indicating which of the half-edges in the encoded tree corresponds to half-edge $h$).
From the definition of $\Pi'$, it follows immediately that $\Pi'$ is equivalent to $\Pi$ in the sense that any solution for $\Pi$ can be transformed (by a deterministic distributed algorithm) in constant time into a solution for $\Pi'$, and vice versa.
Hence, for the purposes of this work, we can safely restrict our attention to node-edge-checkable \lcl{}s.

\section{Rooting a Tree}\label{sec:rooting}

In this section, we describe a novel $O(\log n)$ time deterministic low-space \mpc algorithm that roots a tree, i.e., orients the edges of a graph such that they point towards a unique root, using $O(m)$ words of global memory. This rooting method is an essential subroutine in the high-complexity regime, but it may also be of independent interest. We start by introducing the technique of \textit{path exponentiation}, which is used to contract long paths in logarithmic time in a memory efficient way. By leveraging the fact that in trees, at least half of the nodes are of degree $<3$, one could apply path exponentiation in a straightforward manner to root a tree in $O(\log^2 n)$ time. Our main contribution is pipelining this process, reducing the runtime to $O(\log n)$, and solving numerous small challenges that arise along the way.

\begin{lemma} \label{lemma:rooting}
	Trees can be rooted deterministically in $O(\log n)$ time in the low-space \mpc model using $O(m)$ words of global memory. If the input graph is a forest, the algorithm is component-stable. For a forests consisting of connected components $T_1,T_2,\dots,T_k$, the runtime becomes $O(\log |T_{\max}|)$, where $T_{\max}$ is the largest component.
\end{lemma}

\begin{proof}
	The algorithm and the proof for constant-degree trees is presented in \Cref{subsection:treerooting}, where the algorithm is clearly deterministic. The extension to arbitrary-degree trees is the following. If all edges of a node fit into one machine, nothing changes from the constant-degree case. Otherwise, we can host a (local) broadcast tree (similar to \Cref{sec:broadcasttree}) for every node $v$ with degree $\omega(n^\delta)$. If some node $u$ wants to communicate with $v$, the communication happens with the machine storing edge $\{u,v\}$.
	
	Component-stability and the compatibility with forests is simple to argue about. The only communication between disconnected components happens in the beginning of post-processing, when all components wait until the root has been found in all components. Clearly, this does not affect the resulting rooting in each component. It does however affect the runtime, since smaller components may have to wait until larger components have found the root. Hence, the runtime becomes $O(\log |T_{\max}|)$.

\end{proof}

\subsection{Path exponentiation} \label{sec:pathexp}

Let us introduce the technique of path exponentiation, which is a logarithmic time technique to compress a path such that upon termination, the endpoints share a virtual edge (defined next). The technique is memory efficient in the sense that in addition to the input edges (i.e., edges incident to a node in the input graph), all nodes in a path keep at most two virtual edges in memory. Consider a path $P$ with \textit{endpoints} $s,t$ and \textit{internal nodes} in $P \setminus \{s,t\}$. Leaf node are considered to be endpoints, and degree-2 nodes are consider to be internal nodes. Nodes in $P$ always keep their input edges in memory. Path exponentiation is initialized by duplicating all edges in $P$ and calling this new path the \textit{virtual graph}. Nodes connected by a virtual edge are called virtual neighbors. A new virtual edge $\{v,w\}$ can be created by node $u$ if there previously existed virtual edges $\{u,v\}$ and $\{u,w\}$. In practice, creating a virtual edge $\{v,w\}$ entails node $u$ informing $v$ the ID of $w$ and $w$ the ID of $v$, i.e., node $u$ \textit{connects} nodes $v$ and $w$. Path exponentiation is executed only on this virtual graph. So henceforth, when talking about neighbors and edges, we refer to virtual neighbors and edges, unless specified otherwise.

In each (path) exponentiation step, endpoints and internal nodes are handled separately. During exponentiation, an internal node $u$ has exactly two neighbors and it can be one of three types: (1) neither neighbor is an endpoint (2) one neighbor is an endpoint and one is an internal node (3) both neighbors are endpoints. In each exponentiation step, an internal node $u$ does the following. 

\begin{itemize}
    \item Node $u$ communicates with its neighbors to learn if it is of type 1, 2 or 3
    \item For each node type:
    \begin{enumerate}
        \item Connects its neighbors $v$ and $w$ with an edge and removes edges $\{u,v\}$ and $\{u,w\}$.
        \item Connects its internal neighbor $v$ to its endpoint neighbor $s$ (or $t$) with an edge. It removes the edge $\{u,v\}$, but keeps the edge $\{u,s\}$ (or $\{u,t\}$) in memory.
        \item Connects its endpoint neighbors $s$ and $t$ with an edge and keeps edges $\{u,s\}$ and $\{u,t\}$ in memory.
    \end{enumerate}
\end{itemize}

If we perform the former exponentiation steps as is, both endpoints will aggregate one edge for each node in the path, which may break local and global memory restrictions. To resolve this issue, we implicitly assume the following scheme. If a node is connected to an endpoint, it keeps track if is the furthest away from said endpoint in the input graph, among all nodes that are connected to the endpoint. Immediately after initializing path exponentiation, the furthest away node is the neighbor of the endpoint in the input graph. During exponentiation:

\begin{itemize}
    \item If a node is the furthest away from an endpoint, when creating a new edge between an endpoint and an internal node, it informs both nodes of the new edge. 
    
    \item If a node is not the furthest away from an endpoint, when creating a new edge between an endpoint and an internal node, it only informs the internal node of the new edge. 
    
    \item If a node is an endpoint, upon receiving a new edge, it drops the old one. 
\end{itemize}

This scheme results in endpoints $s$ and $t$ effectively doing nothing during path exponentiation, except keeping track of the latest edge connecting them to an internal node. Eventually, exponentiation terminates when $s$ and $t$ get connected and all internal nodes have two edges, one for each endpoint. As the shortest distance between $s$ and $t$ in the virtual graph decreases by at least a factor of $3/2$ in each step, path exponentiation terminates in $O(\log n)$ time. Due to the aforementioned memory saving scheme, all nodes in $P$ keep at most two edges in memory, resulting in $O(m)$ global memory use.

\begin{figure}
    \centering
    \includegraphics[width=0.6\textwidth]{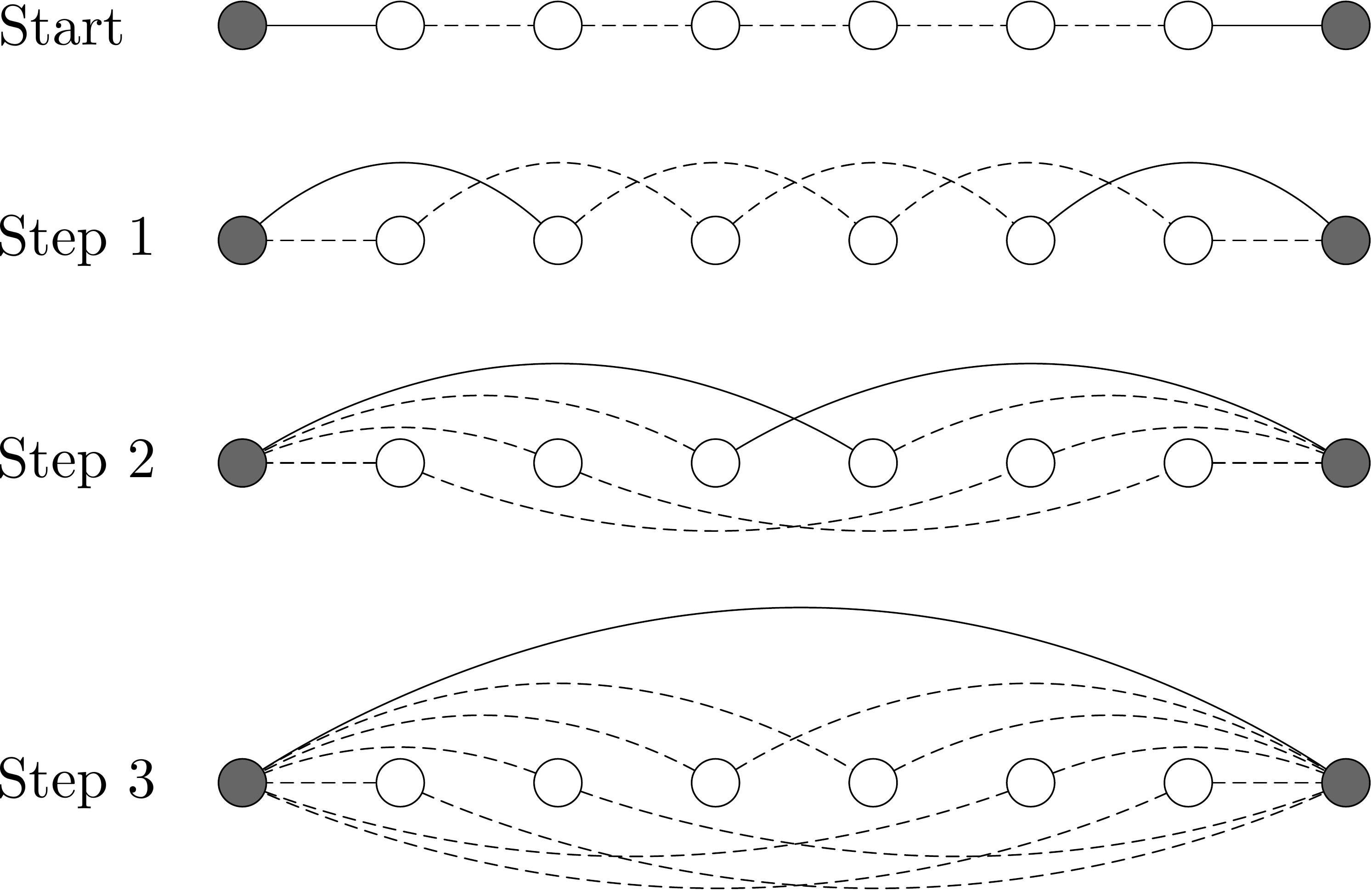}
    \caption{Path exponentiation on a path with 8 nodes. All edges are virtual, and solid edges emphasize the edges that endpoints keep track of. Exponentiation terminates in 3 steps, after which, internal nodes are connected to both endpoints and endpoints are connected by an edge.}
    \label{fig:effexp}
\end{figure}

\subsection{Tree Rooting Algorithm} \label{subsection:treerooting}

The algorithm is split into two parts. First we find the root node (\Cref{subsubsec:findroot}), during which we set unoriented paths aside. Then, during post-processing (\Cref{subsubsec:orientingpaths}), we orient said paths in parallel, resulting in a total runtime of $O(\log n)$, local memory use $O(n^\delta)$, and global memory use $O(m)$. We emphasize that we execute the following algorithms on the virtual graph, and not on the input graph. Initially, the virtual graph is an identical copy of the input graph. Also, when talking about neighbors and edges, we refer to virtual neighbors and edges, unless specified otherwise.

\subsubsection{Finding the Root} \label{subsubsec:findroot}

The high level idea is quite simple: perform path exponentiation in all current paths (note that now, endpoints are either leaf nodes or nodes of degree $\geq 3$), and when an endpoint of a path is a leaf that is connected to the other endpoint, we set the path aside (this will become apparent later). For this to work, one major issue must be addressed. An endpoint of degree $\geq 3$ can turn into a degree-$2$ node, extending the current path. The difficulty in this scenario stems from the fact that some nodes are in the middle of path exponentiation and some have not yet started. We resolve this issue by defining nodes that were endpoints of degree $\geq 3$ in the previous phase, but are nodes of degree $2$ in the current phase, as \textit{midpoints}.

\begin{enumerate}
    \item In phase $i$, each node $u$ in a path first identifies if it is an endpoint, a midpoint, or an internal node.
    
    \begin{itemize}
        \item If $u$ is a leaf node that is connected to the other endpoint, we set the path containing $u$ aside$^*$. Note that this also applies to paths of length 1.
        
        \item If $u$ is an endpoint that is not connected to the other endpoint, it does nothing except act as endpoint for the corresponding internal nodes.
        
        \item If $u$ is a midpoint such that both of its neighbors are other midpoints or endpoints, it transforms into an internal node and acts as such henceforth. 
        
        \item If $u$ is a midpoint such that at least one of its neighbors is an internal node, it does nothing except act as endpoint node for the corresponding internal nodes. 
        
        \item If $u$ is an internal node (or a midpoint that has turned into an internal node), it performs path exponentiation.
    \end{itemize}
    
    $^*$Setting path $P$ aside entails leaf node $s$ informing the other endpoint $t$ that the orientation is going to be from $s$ to $t$, so that the algorithm can proceed. The internal nodes of the path do not need to be informed that they are set aside, since they will not do anything for the remainder of the algorithm. Also, instead of $P$, we are actually setting aside $P \setminus t$, since $t$ may be of high degree and has to remain in the graph. Observe that this means that both endpoints of the path we are setting aside are leaves that know the orientation of the path. The edges of the paths remain unoriented until the root is found, after which these paths are oriented in parallel during the post-processing in \Cref{subsubsec:orientingpaths}. Note that if both endpoints of a path are leaves, the algorithm terminates and the higher ID node is picked to be the root.
\end{enumerate}

\paragraph*{Correctness.} Since we only orient paths connected to at least one leaf node, we end up with a valid orientation and a unique root.

\paragraph*{Runtime.} Consider endpoints $s$ and $t$ of some path during some phase. The aim of the algorithm is essentially to construct edge $\{s,t\}$. Now consider a current shortest (virtual) path $P_v$ between $s$ and $t$.  Observe that the nodes responsible for eventually creating edge $\{s,t\}$ constitute $P_v$. Hence, all other nodes are redundant and can be thought of as removed.

In order to analyze the number of nodes that are removed in a phase, we want to first count the number of internal nodes in paths such as $P_v$. Since a midpoint is always incident to an internal node (otherwise it would transform into an internal node), at least $1/3$ of all nodes in $P_v$ are internal nodes. This is evident from the ``worst case'' where two consecutive midpoints are followed by one internal node.

Since all internal nodes in $P_v$ perform path exponentiation, the number of internal nodes in $P_v$ drops by a factor of at least $3/2$ in one phase. Observe that in addition to removing at least $1/3$ of the internal nodes in all paths such as $P_v$, we also remove all leaf nodes. 

Since the average degree of a node in a tree is $<2$, the number of leaves in a tree is larger than the number of nodes of degree $\geq 3$. Hence, in each phase, we remove at least $1/3 \cdot 2/3 = 2/9$ of all of the nodes in the graph, and the algorithm finds the root after $O(\log n)$ phases.

\paragraph*{Memory.} The only memory usage stems from path exponentiation, where in each path, in addition to the input edges (i.e., edges incident to a node in the input graph), all nodes keep at most two virtual edges in memory. Observe that endpoints can partake in multiple path exponentiations. However, since endpoints keep track of only one virtual edge (per path exponentiation), it is easy to see that an endpoint can never have more virtual edges than input edges. Hence, local memory $O(n^\delta)$ and global memory $O(m)$ are not violated.

\subsubsection{Post-processing} \label{subsubsec:orientingpaths}

Before initializing this part, we first we have to ensure that the root finding has terminated, which can be done using the broadcast tree (\Cref{sec:broadcasttree}) in constant time. Then, we can start orienting the paths that were set aside by the root finding algorithm. Recall that they are paths where both endpoints are leaves that know the orientation. We want to orient all edges in these (possibly very long) paths in parallel.

\begin{enumerate}
    \item Consider performing path exponentiation on a path $P$ such that when an edge is created between an endpoint and an internal node, it is oriented according to the orientation information at the endpoint. Upon termination, all nodes orient their input edges according to the orientation of their virtual edges. Note that this requires nodes to keep track which virtual edge corresponds to which edge in the input graph.
\end{enumerate}

\paragraph*{Correctness, Runtime, Memory.} Observe that an oriented edge is created only by nodes that already have an oriented edge (are of type 2 or 3 in \Cref{sec:pathexp}) and hence, the orientation will be correct. As we only perform path exponentiation, the runtime is $O(\log n)$. Clearly, this only has a constant overhead compared to path exponentiation. Hence, local memory $O(n^\delta)$ and global memory $O(m)$ are not violated.

\section{The High-Complexity Regime}\label{sec:high}

In this section, we will prove that all solvable \lcl problems on trees, i.e., \emph{all \lcl problems that have a correct solution on every considered tree}, can be solved deterministically in $O(\log n)$ time in the low-space \mpc model using $O(m)$ words of global memory. Our proof is constructive: we explicitly provide, for any solvable \lcl, an algorithm that has a runtime of $O(\log n)$. In fact, our construction can be used to find an $O(\log n)$ time algorithm $\fA$ \emph{even for unsolvable \lcl{}s}, with the guarantee that on any instance that admits a correct solution the given output will be correct. Note that we can easily extend $\fA$ to output ``no solution'' for instances that do not admit a correct solution, by checking, after executing $\fA$, around each node whether the output satisfies the \lcl constraints. If for some node the \lcl constraints are not satisfied, it can broadcast ``no solution'' to all nodes using the broadcast tree of \Cref{sec:broadcasttree}.The following theorem follows by \Cref{cor:fewernodes,lem:runtimephaseone,lem:runtimephasetwo,lem:gi,lem:wellandcorrect}, and the discussion in \Cref{sec:highimple}.

\begin{theorem} \label{thm:high}
	Let $\Pi$ be an \lcl problem in constant-degree trees that admits a correct solution. There is an $O(\log n)$ time deterministic low-space \mpc algorithm for $\Pi$ that uses $O(m)$ words of global memory.
\end{theorem}
    
\subsection{High-level Overview}\label{sec:overviewhcr}
Consider an arbitrary \lcl problem $\Pi$ that has a correct solution on all considered instances.
Throughout this section, we will assume that the \lcl is given as a node-edge-checkable \lcl (\Cref{def:nodeedge}), which we can do w.l.o.g., as observed in \Cref{sec:prelim}. In the following, we will give a slightly simplified view of the algorithm $\fA$ we will use to solve $\Pi$ in $O(\log n)$ time. On a high level, algorithm $\fA$ proceeds in $3$ phases. The first phase consists in rooting the input tree, by using the algorithm described in \Cref{sec:rooting}.

In the second phase, which we will refer to as the \emph{leaves-to-root phase}, roughly speaking, the goal is to compute, for a substantial number of edges $e = (u, v)$, the set of output labels that can be output at half-edge $(v, e)$ such that the label choice can be extended to a (locally) correct solution in the subtree hanging from $v$ via $e$.
This is done in an iterative manner, proceeding from the leaves towards the root. When, at last, the root has computed this set of output labels for each incident half-edge, it can, on each such half-edge, select an output label from the computed set such that the obtained node configuration is contained in the node constraint of $\Pi$ and the input-output constraints of $\Pi$ (given by the function $\gee$ in the definition of $\Pi$) are satisfied. Such a selection must exist due to the fact that $\Pi$ has a correct solution on the considered instance.

The last phase, which we will refer to as the \emph{root-to-leaves phase}, consists in completing the solution from the root downwards, by iteratively propagating the selected solution further towards the leaves.
With the same argumentation as at the root, certain nodes $v$ can select an output label at the half-edge leading to its parent and output labels from the sets computed on its incident half-edge leading to its children such that the obtained node configuration is contained in the node constraint of $\Pi$, the obtained edge configuration on the edge from $v$ to its parent is contained in the edge constraint of $\Pi$, and the input-output constraints of $\Pi$ are satisfied.

Unfortunately, the depth of the input tree can be $\omega(\log n)$, which prevents us from performing the outlined phases in a sequential manner (if we want to obtain an $O(\log n)$ time algorithm). In order to mitigate this issue, we will not only process the leaves of the remaining unprocessed tree in each iteration, but also the nodes of degree $2$, inspired by the rake-and-compress decomposition which guarantees that after $O(\log n)$ iterations of removing all degree-$1$ and degree-$2$ nodes all nodes have been removed. The advantage of degree-$2$ nodes over higher-degree nodes w.r.t.\ storing completability information (as in the above outline) is that they form paths, which by definition only have two endpoints; the idea, when processing such a path, is to simply store in the two endpoints the information for which pairs of labels at the two half-edges at the ends of the path there exists a correct completion of the solution inside the path. This allows to naturally add processing degree-$2$ nodes to the leaves-to-root phase of the algorithm outline provided above, while for the root-to-leaves phase, the information stored at the endpoints $s, t$ of a path (where $s$ is an ancestor of $t$) essentially allows us to start extending the current partial solution on the path itself (and thereafter on the subtrees hanging from nodes on the path) one step after the output labels at $s$ and $t$ are selected.

However, in the leaves-to-root phase, even when using graph exponentiation, processing a path of degree-$2$ nodes of length $L$ takes $\Omega(\log L)$ time, whereas the $O(\log n)$ time guarantee of the rake-and-compress technique crucially relies on the fact that each iteration can be performed in constant time.
Hence, essentially, we will only perform one step of graph exponentiation on paths in each iteration.
Here, a new obstacle arises: before the graph exponentiation is finished, new nodes (that just became degree-2 nodes due to all except one of their remaining children being conclusively processed in the most recent iteration) might join the path.
Nevertheless, we will show that this process still terminates in logarithmic time by designing a fine-tuned potential function that is inspired by the idea of counting how many nodes from certain groups of degree-$2$ nodes are contained in any fixed ``pointer chain'' from some leaf to the root. This pointer chain way of thinking is similar to the path exponentiation technique introduced in \Cref{sec:pathexp}, which solves a similar problem. 

Another issue is that we have to be able to store the completability information that we compute in the leaves-to-root phase until we use it (again) in the root-to-leaves phase.
As the number of new edges/pointers introduced by the graph exponentiation on paths can be up to logarithmic in $n$ per node (even on average), a naive implementation will result in a logarithmic overhead in the used memory.
In order to remedy this problem, we split the leaves-to-root phase (and, as a result thereof, also the root-to-leaves phase) into two subphases.
While the second subphase proceeds as explained above, the first subphase differs by processing the degree-$2$ nodes, i.e., paths, in a way that guarantees that the number of new edges introduced by the graph exponentiation (which we should rather call pointer forwarding at this point) on each path in each iteration is only a constant fraction of the length of the respective path.
This is achieved by finding, in each iteration, an MIS on each path, letting only MIS nodes forward pointers, and removing the MIS nodes afterwards.
The first subphase consists of $\Theta(\log \log n)$ iterations; we will show that this ensures that the overhead introduced by the MIS computations does not increase the overall asymptotic runtime, and that the number of remaining nodes is in $O(n/\log n)$.
The latter property ensures that the memory overhead of $O(\log n)$ edges per node introduced in the second subphase does not exceed the desired global memory of $O(m)$ words.
Lastly, in \Cref{sec:highimple} we take care of the local memory restrictions.

    
\subsection{The Algorithm}\label{sec:highalgo}
In this section, we provide the desired algorithm that can be implemented in $O(\log n)$ time in the low-space \mpc model and prove its correctness. 
The details about the precise implementation in the \mpc model are deferred to \Cref{sec:highimple}.

Let $\Pi = (\sinn, \sout, \noco, \edco, \gee)$ be the considered \lcl, and let $G$ denote the input tree.
Before describing the algorithm, we need to introduce the new notion of a \emph{compatibility tree}.
In a sense, a compatibility tree can be thought of as a realization of an \lcl on a given input tree where the constraints that two labels on an edge or $\deg(v)$ labels around a node $v$ have to satisfy (as well as which output labels can be used at which half-edge, which the \lcl encodes via input labels) are already encoded on the edge and around the node.

\begin{definition}\label{def:comptree}
    A \emph{compatibility tree} is a rooted tree $T$ (without input labels)
    where each edge $(u,v)$ is labeled with a subset $S_{uv}$ of $\sout \times \sout$, and each node $w$ is labeled with a tuple $S_w$ consisting of tuples of the form $(s_w^e)_{e \in \incid(w)}$ where $\incid(w)$ denotes the set of edges incident to $w$, and $s_w^e \in \sout$ for each $e \in \incid(w)$.
    A \emph{correct solution} for a compatibility tree is an assignment $\gout \colon H(T) \to \sout$ such that
    \begin{enumerate}
        \item for each edge $e = (u,v)$, we have $(\gout((u,e)), \gout((v,e))) \in S_{uv}$, and
        \item for each node $w$, there exists a tuple $(s_w^e)_{e \in \incid(w)} \in S_w$ such that, for each edge $e \in \incid(w)$, we have $\gout((w,e)) = s_w^e$.
    \end{enumerate}
\end{definition}

Now, we are set to describe the desired algorithm.
The algorithm starts by rooting $G$, using the method described in \Cref{sec:rooting}.
We denote the root by $\rooot$.
Then, we transform $G$ (which from now on will denote the rooted version of the input tree) into a compatibility tree $G'$ by iteratively removing nodes of degree $1$ and $2$ (while suitably updating the edge set) and assigning a subset $S_{uv}$, resp.\ $S_w$, to each remaining edge $(u,v)$, resp.\ remaining node $w$ (see \Cref{sec:trafothere}).
Next, we find a correct solution for $G'$ (see \Cref{sec:highmainalgo}), and finally we transform the obtained solution into a correct solution for \lcl $\Pi$ on $G$ (see again \Cref{sec:trafothere}).

\subsubsection{Reducing the \lcl to a Compatibility Tree with Fewer Nodes}\label{sec:trafothere}
In this section, we show how to transform the rooted tree $G$ into a compatibility tree with $O(n/\log n)$ nodes, where $n$ is the number of nodes of $G$, and how to transform any correct solution for $G'$ into a correct solution for the given \lcl $\Pi$ on $G$.
In other words, we show how to reduce the problem of solving $\Pi$ on $G$ to the problem of finding a correct solution for a compatibility tree with fewer nodes.
The idea behind this approach is that the new, smaller instance can be solved in logarithmic time without exceeding the desired global memory of $O(m)$ words.
To obtain a good overall runtime, we will also show in this section how to perform the reduction (and recover the solution) in $O(\log \log n \cdot \log^* N) = O(\log \log n \cdot \log^* n)$ rounds, modulo some implementation details left to \Cref{sec:highimple}.
Recall that $N=\text{poly}(n)$ denotes the size of our ID space.

We start by describing how to obtain $G'$ from $G$.
To this end, we will first transform $G$ into a compatibility tree $G_0$, and then iteratively derive a sequence $G_1, G_2, \dots, G_t$ of compatibility trees from $G_0$, where $t \in O(\log \log n)$ is a parameter we will choose later.

We define $G_0$ in the natural way, by essentially encoding the given \lcl $\Pi$.
The nodes and edges of $G_0$ are precisely the same as in $G$.
For any edge $(u,v)$ with input labels $\iota_u$ and $\iota_v$ at the two half-edges belonging to $(u,v)$, we set $S_{uv}$ to be the set of all pairs $(\ell, \ell') \in \sout \times \sout$ such that the multiset $\{ \ell, \ell' \}$ is contained in the edge constraint $\edco$ of $\Pi$, and we have $\ell \in \gee(\iota_u)$ and $\ell' \in \gee(\iota_v)$.
For any node $w$, we set $S_w$ to be the set of all tuples $(\ell^e)_{e \in \incid(w)}$ such that the multiset $\{ \ell^e \mid e \in \incid(w) \}$ is contained in $\noco_{\deg(w)}$.
Note that the asymmetric nature of this definition comes from the fact that we only need to require compatibility with the function $\gee$ once, in the constraints for nodes or (as we chose) for edges.
From the definition of $G_0$, we obtain directly the following observation.

\begin{observation}\label{obs:equiv}
    A half-edge labeling of $G$ is a correct solution for \lcl $\Pi$ if and only if it is a correct solution for the compatibility tree $G_0$ (under the natural isomorphism between $G$ and the graph underlying $G_0$).
\end{observation}

We now describe how to obtain $G_i$ from $G_{i-1}$, for any $1 \leq i \leq t$.
We transform $G_{i-1}$ into $G_i$ in two steps.

In the first step, we start by finding an MIS $Z$ on the subgraph of $G_{i-1}$ induced by all nodes of degree precisely $2$.
Then, for each node $v \in Z$ with incident edges $e = (u,v)$ and $e' = (v,w)$, we remove $e$ and $e'$ from $G_{i-1}$ and replace them by a new edge $e'' = (u,w)$.
Furthermore, for the new edge, we set $S_{uw}$ to be the set of all label pairs $(\ell, \ell') \in \sout \times \sout$ such that there exist labels $\ell_1, \ell_2$ satisfying $(\ell, \ell_1) \in S_{uv}$, $(\ell_2, \ell') \in S_{vw}$, $s^e_v = \ell_1$, and $s^{e'}_v = \ell_2$.
From the perspective of the nodes $u$ and $w$, the new edge $e''$ replaces the old edges $e$ and $e'$, respectively, in the indexing hidden in the definition of the tuples $S_u$ and $S_w$.
Let us call the obtained graph $G'_{i-1}$.

In the second step, executed after the first step has finished, each edge $e^* = (x,y)$ such that $x$ is a leaf is removed together with $x$.
Moreover, for such a removed edge, we set $S_y$ to be the set of all tuples $(\ell^e_y)_{e \in \incid(y)}$ such that there exist labels $\ell, \ell'$ such that (in $G'_{i-1}$) we have $(\ell, \ell') \in S_{xy}$ and there exists some tuple $(s^e_y)_{e \in \incid(y)}$ with $s^e_y = \ell^e_y$ (for all $e \neq e^*$) and $s^{e^*}_y = \ell'$.
(Here the first occurrence of $\incid(y)$ denotes the set of edges incident to $y$ after removing $e^*$, while the second occurrence denotes the set before removing $e^*$.)
If a node $y$ of $G'_{i-1}$ has multiple children that are leaves, then we can think of removing the respective edges one by one, each time updating $S_y$.
However, for the actual computation, node $y$ can perform all of these steps at once.

We obtain the following lemma.

\begin{lemma}\label{lem:gi}
    Let $1 \leq i \leq t$.
    If there exists a correct solution for $G_{i-1}$, then there also exists a correct solution for $G_i$.
    Moreover, given any correct solution for $G_i$, we can transform it into a correct solution for $G_{i-1}$ in a constant number of rounds in the low-space MPC model using $O(m)$ words of global memory.
    Finally, given $G_{i-1}$, we can compute $G_i$ in $O(\log^* N) = O(\log^* n)$ rounds in the described setting.
\end{lemma}
\begin{proof}
    The first statement follows directly from the definition of $G_i$.
    For the second statement, observe that from the definition of $G_i$, it follows that any correct solution for $G_i$ provides a partial solution for $G_{i-1}$ (under the natural transformation that subdivides edges and adds the ``removed'' leaves with their incident edges) that is part of a correct solution for $G_{i-1}$ (and this factors through $G'_{i-1}$ in the obvious way).
    Hence, we can first obtain a correct solution for $G'_{i-1}$ by extending the provided solution on the removed leaves with their incident edges, and then obtain a correct solution for $G_{i-1}$ by doing the same on the subdivided edges.
    Note that the first extension can be performed by the nodes that are incident to the leaves (all of which have only constantly many output labels to determine), and the second extension by the computed nodes in the MIS $Z$ (which we can do in parallel since no two nodes in $Z$ are neighbors).
    The third statement follows from the definition of $G_i$, the fact that an MIS can be computed in $O(\log^* N) = O(\log^* n)$ rounds (already in the \local model), and the above observation about parallelization.
\end{proof}

Next, we bound the number of nodes of $G_i$, which we denote by $n_i$.

\begin{lemma}\label{lem:rcanalysis}
    For any $1 \leq i \leq t$, we have $n_i \leq 2/3 \cdot n_{i-1}$.
\end{lemma}
\begin{proof}
    By the construction of $G_i$, all nodes that are contained in $G_{i-1}$ but not in $G_i$ are either leaves or degree-$2$ nodes in $G_{i-1}$.
    In particular, as the set of leaves is the same in $G_{i-1}$ and $G'_{i-1}$, all leaves of $G_{i-1}$ are not contained in $G_i$.
    Regarding degree-$2$ nodes in $G_{i-1}$, we observe that at least a third of them must be part of the chosen MIS since (i) each degree-$2$ node must be in the MIS or have an MIS node as neighbor, and (ii) each MIS node covers at most three nodes (in the sense that it is equal or adjacent to them).
    
    Also, since the average degree of a node in a tree is below $2$, the number of leaves in a tree is larger than the number of nodes of degree at least $3$. Hence, when going from $G_{i-1}$ to $G_i$, at least half of the nodes of degree $\neq 2$ are removed, and in total we obtain that the number of nodes that are removed is at least $1/3 \cdot n_i$, which proves the lemma.
\end{proof}

Now, by setting $t \coloneqq 2 \log \log n$ and $G' \coloneqq G_t$, we obtain the following straightforward corollary.
\begin{corollary}\label{cor:fewernodes}
    The number of nodes of $G'$ is at most $n/\log n$.
\end{corollary}

Moreover, by \Cref{obs:equiv,lem:gi,lem:rcanalysis}, we know that we can compute $G'$ in $O(\log \log n \cdot \log^* N) = O(\log \log n \cdot \log^* n)$ rounds, that there is a correct solution for $G'$ (provided the \lcl $\Pi$ admits a correct solution on $G$), and that we can transform any correct solution for $G'$ into a correct solution for $\Pi$ on $G$ in $O(\log \log n)$ rounds.
The stated runtimes are under the premise that we can implement all of the $O(\log \log n)$ steps without running into memory issues, which we will show to be the case in \Cref{sec:highimple}.
(Note that \Cref{lem:gi} only makes statements about single steps.)

\subsubsection{Solving the Compatibility Tree}\label{sec:highmainalgo}
In this section, we design an algorithm $\fA'$ that computes a correct solution for the compatibility tree $G'$ computed in \Cref{sec:trafothere}.
Algorithm $\fA'$ proceeds in two phases that largely correspond to the leaves-to-root and root-to-leaves phases mentioned in the high-level overview.

\paragraph{Phase I}\label{sec:phaseone}
In Phase I, we maintain a set of pointers $(u,v)$, which, roughly speaking, encode which output labels can be chosen at $u$ and $v$ such that the solution can be correctly completed on the path between $u$ and $v$ and the subtrees hanging from this path.
The goal is to increase the lengths of these pointers, until we obtain leaf-to-root pointers that allow us (in Phase II) to fix output labels at the root that can be completed to a correct solution on the whole tree.
The algorithm in Phase I proceeds in iterations $i = 1, 2, \dots$.
Before explaining the steps taken in each iteration, we need to introduce some definitions.

A pointer is simply a pair $(u, v)$ of nodes such that $v$ is an ancestor of $u$ in $G'$, i.e., such that $v$ is a node on the path from $u$ to the root $\rooot$, and $v \neq u$.
We say that a pointer $(u, v)$ \emph{starts} in $u$ and \emph{ends} in $v$; we also call $(u, v)$ an \emph{incoming} pointer when considering node $v$, and an \emph{outgoing} pointer when considering node $u$. 
On each pointer $p = (u, v)$, we store several pieces of information which are required for Phase II:
\begin{enumerate}
    \item a set $\pairs_p \subseteq \sout \times \sout$ (encoding completability information as outlined above),
    \item a node $\pred_p$ which might also be empty, i.e., $\pred_p \in V(G') \cup \{\bot\}$ (encoding the information about which node created the pointer), and
    \item the pair $(\first_p, \last_p)$ where $\first_p$ is the first edge and $\last_p$ the last edge on the unique path from $u$ to $v$ (possibly $\first_p = \last_p$).
\end{enumerate}

The initial pointer set $\ps_0$ is set to $E(G)$, where for each pointer $p = (u,v)$, we set $\pairs_p \coloneqq S_{uv}$, $\pred_p = \bot$, and $\first_p = \last_p = (u,v)$.
We also maintain a set of \emph{active pointers} which is initially set to $\aps_0 \coloneqq \ps_0$.
Throughout Phase I, we will guarantee that the set of active pointers contains, for each node $u \neq \rooot$, at most one pointer starting in $u$, and no pointers starting in the root $\rooot$.
We call a node \emph{active} in iteration $i$ if it is the root $\rooot$ or has exactly one outgoing active pointer at the end of iteration $i - 1$, i.e., one outgoing pointer in $\aps_{i-1}$.
For active nodes, we will denote the unique pointer starting in node $u \neq \rooot$ by $p(u)$.
Finally, for each node $u$ that is not a leaf, we also maintain a tuple $\comp(u) = (\comp^e(u))_{e \in \incid(u)}$ such that each element $\comp^e(u)$ is either a subset of $\sout$ or the special label $\undec$.
(The purpose of these tuples is to store completability information about subtrees hanging from $u$ via different edges).
The multiset for node $u$ is initially set to $\comp_0(u) \coloneqq (\undec, \dots, \undec)$. 

In iteration $i$, the pointer set is updated from $\ps_{i-1}$ to $\ps_i$, the active pointer set is updated from $\aps_{i-1}$ to $\aps_i$, and, for each non-leaf node $u$, the tuple $\comp_{i-1}(u)$ is updated to $\comp_i(u)$.
For each $i \geq 0$, we will ensure that $\aps_i \subseteq \ps_i$ and $\ps_i \subseteq \ps_{i+1}$ (if both are defined).
In order to specify the precise update rules, we need two further definitions.

The first definition specifies which nodes (during the iterative process explained above) should be intuitively regarded as degree-$2$ nodes (because for all incident edges (to children) except one, the corresponding subtree has already been completely processed w.r.t.\ completability information) and which as nodes of degree at least $3$.
(Degree-$1$ nodes will only play a passive role in the update rules.)
\begin{definition}[$k$-nodes]\label{def:123node}
    We call a node $u \neq r$ a $2$-node if $u$ has at least one incoming active pointer and for any two incoming active pointers $p, p'$, we have $\last_p = \last_{p'}$.
    For a $2$-node $u$, we call the unique incoming edge $e$ satisfying that for any incoming active pointer $p$ we have $\last_p = e$, the \emph{relevant in-edge} of $u$.
    We call a node $u \neq r$ a $3$-node if $u$ has (at least) two incoming active pointers $p, p'$ satisfying $\last_p \neq \last_{p'}$.
    For a node $u$ that is a $3$-node or the root $\rooot$, an incoming edge $e$ is called a \emph{relevant in-edge} of $u$ if there is an incoming active pointer $p$ satisfying $\last_p = e$.
    We call a node $u \neq r$ a $1$-node if $u$ has no incoming active pointer.
\end{definition}

The second definition provides an operation that combines two pointers into a larger one.
\begin{definition}[Merge]\label{def:merge}
    Let $v$ be a $2$-node, $p = (u, v)$ and $p' = (v, w)$ two active pointers starting and ending in $v$, respectively, and $\comp(v) = \{ \comp^e(v) \}_{e \in \incid(v)}$ the current tuple at $v$.
    (Note that our construction of the update rules will guarantee (as shown in \Cref{obs:undec}) that (for any $2$-node $v$) we have $\comp^e(v) = \undec$ if and only if $e = \last_p$ or $e = \first_{p'}$.)
    Then, we set $\merge(p, p') \coloneqq (u, w)$.
    Furthermore, we set $\pairs_{\merge(p, p')}$ to be the set of all label pairs $(\ell, \ell')$ such that there exists a tuple $(\ell^e)_{e \in \incid(v)}$ of output labels such that
    \begin{enumerate}
        \item $(\ell^e)_{e \in \incid(v)} \in S_v$,
        \item $\ell^e \in \comp^e(v)$, for each $e \in \incid(v) \setminus \{\last_p, \first_{p'}\}$, and
        \item $(\ell, \ell^{\last_p}) \in \pairs_p$ and $(\ell^{\first_{p'}}, \ell') \in \pairs_{p'}$.
    \end{enumerate}
    Finally, we set $\pred_{\merge(p, p')} \coloneqq v$, $\first_{\merge(p, p')} \coloneqq \first_p$, and $\last_{\merge(p, p')} \coloneqq \last_{p'}$.
\end{definition}

Note that the above definition ensures that the new pointer $(u, w)$ satisfies the property that $w$ is an ancestor of $u$, i.e., $(u, w)$ is indeed a pointer.
Now, we are set to define precisely how the sets and tuples we maintain throughout the iterations change.
The update rules for iteration $i\geq 1$ are as follows.
We emphasize that update rules~\ref{step2} and~\ref{step3} are executed ``in parallel'' for all $2$-nodes, $3$-nodes, and the root, i.e., there are no dependencies between these steps.

\begin{enumerate}
    \item Start by setting $\aps_i \coloneqq \ps_i \coloneqq \emptyset$, and $\comp_i(u) \coloneqq \comp_{i-1}(u)$ for each node $u$.
    \item\label{step2} For each active $2$-node $u$ (with outgoing pointer $p(u)$), and each incoming active pointer $p \in \aps_{i-1}$, add the pointer $\merge(p, p(u))$ to $\aps_i$.
    \item\label{step3} For each node $u$ that is a $3$-node or the root $\rooot$ do the following.
    Start by asking, for each relevant in-edge $e = (v, u)$ of $u$, whether there is some incoming active pointer $p = (w, u) \in \aps_{i-1}$ such that $\last_p = e$ and $w$ is a leaf.
    If the answer is ``no'' for at least one relevant in-edge, or if $u = r$, then
    \begin{enumerate}
        \item\label{itema} for all relevant in-edges $e$ for which the answer is ``no'', add all incoming active pointers $p' \in \aps_{i-1}$ with $\last_{p'} = e$ to $\aps_i$, and
        \item\label{itemb} for all relevant in-edges $e = (v, u)$ for which the answer is ``yes'', (only) change $\comp^e_i(u)$ (from $\undec$) to the set of all labels $\ell$ satisfying that there exists a pair $(\ell', \ell'') \in \pairs_p$ with $\ell'' = \ell$ and $(\ell') \in S_w$.
    \end{enumerate}
    If the answer is ``yes'' for all relevant in-edges and $u \neq r$, then change the answer to ``no'' on precisely one arbitrarily chosen relevant in-edge, and proceed as in the previous case (i.e., execute steps~\ref{itema} and~\ref{itemb}).
    \item For each node $u \neq \rooot$ that has an outgoing active pointer, set $p(u)$ to be the unique pointer in $\aps_i$ starting in $u$.
    \item Set $\ps_i \coloneqq \ps_{i-1} \cup \aps_i$.
\end{enumerate}

Phase I terminates after the first iteration $i$ satisfying $\aps_i = \emptyset$.
In order to illustrate the algorithm in Phase I, we provide an example in the following.

\begin{figure}
    \centering
    \includegraphics[width=0.88\textwidth]{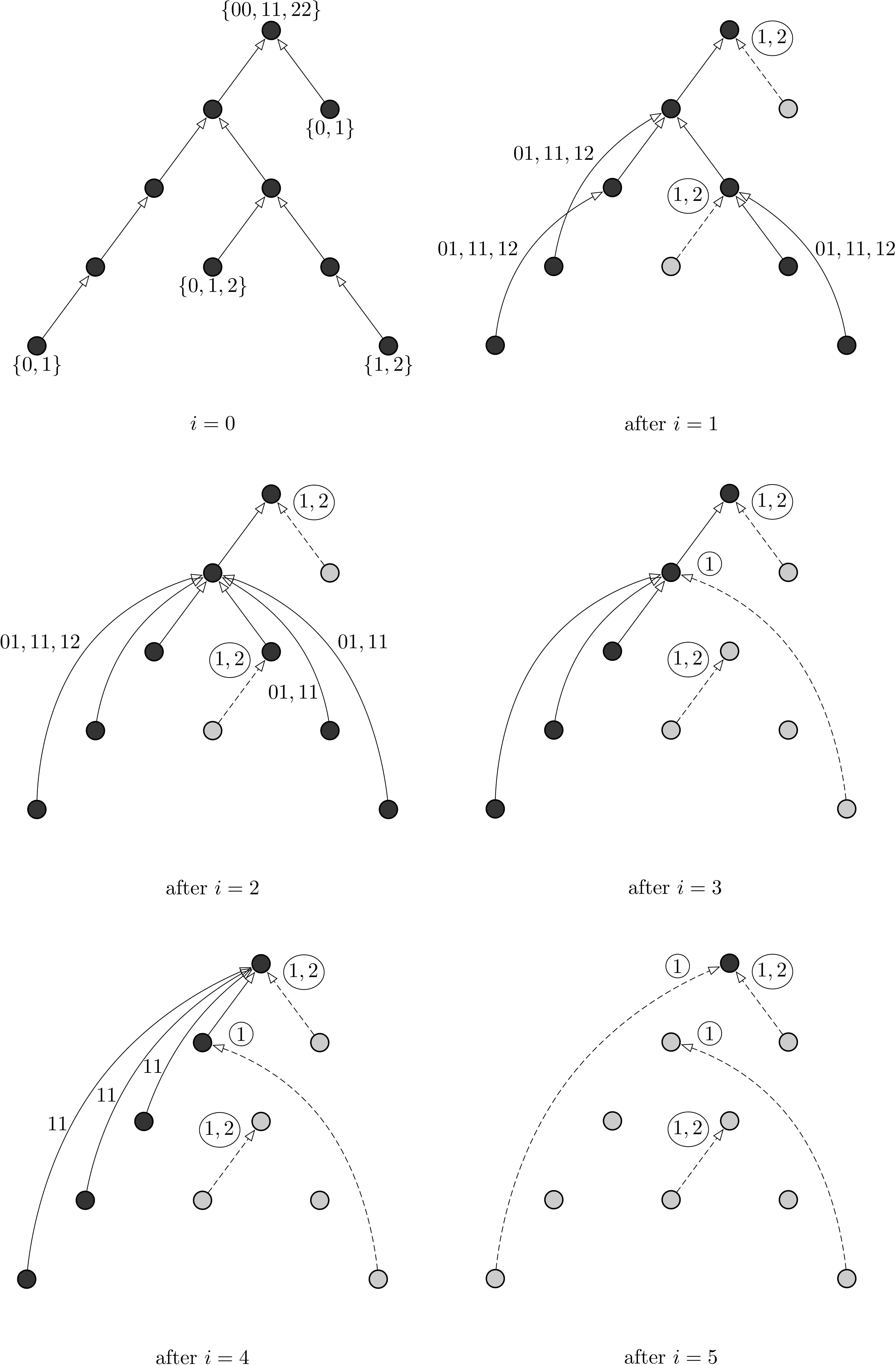}
    \caption{An example execution of the algorithm in Phase I.}
    \label{fig:pointers}
\end{figure}

\paragraph*{Example.}
Consider the compatibility tree given in \Cref{fig:pointers}, for $i=0$.
The considered output label set is $\sout = \{ 0, 1, 2 \}$.
For each edge $(u,v)$ in the tree, the associated label pair set $S_{uv}$ is defined as ${(0,1), (1,2)}$, which, for simplicity, we will write in the form $01, 12$.
For each node $w$ that is not a leaf or the root, the associated set $S_w$ is defined as the set of all tuples (of length $\deg(w)$) with pairwise distinct entries (e.g., if $w$ is of degree $3$, each tuple in $S_w$ is a permutation of $(0,1,2)$).
For the leaves and the root, the associated tuple set is given below or above the node.

In iteration $i = 1$, three merge operations are performed by the three non-root nodes of degree $2$.
In each merge operation, two pointers $p = (u,v)$, $p' = (v,w)$, each labeled with $01, 12$, are merged into a larger pointer $(u,w)$ with label $01, 11, 12$.
The label of the new pointer contains, e.g., the pair $12$ since labeling half-edge $(u,(u,v))$ with $1$ and half-edge $(w,(v,w))$ with $2$ can be completed to a labeling of all inbetween half-edges such that the labeling respects the constraints $S_{uv}$, $S_v$, and $S_{vw}$, namely by labeling $(v,(u,v))$ with $2$, and $(v,(v,w))$ with $1$.
Moreover, the two dashed pointers are removed from the set of active pointers in iteration $1$ as they start in a leaf $x$ and end in a $3$-node or the root.
Consequently, the two leaves in which the pointers start become inactive in iteration $1$ (which is illustrated by coloring them gray).
Furthermore, the two nodes $y$ in which the pointers end update their set $\comp(y)$, by setting the entry corresponding to the removed pointer from $\undec$ to the set of all labels that (when written at half-edge $(y,(x,y))$) are completable downwards, i.e., for which a label at the respective leaf $x$ exists that respects $S_x$ and $S_{xy}$.
We illustrate the entries that change from $\undec$ to some set by circling them, e.g., the changed entry of $\comp(\rooot)$ is the circled set consisting of the labels $1$ and $2$.

In iteration $i=2$, three merge operations are performed, by the two non-root $2$-nodes.
In iteration $i=3$, no merge operations are performed as there is no $2$-node.
Furthermore, the left child $z$ of the root (which is a $3$-node) obtains ``yes'' as answer for the question it asks in step~\ref{step3} of the update rules, for both relevant in-edges.
Thus, $z$ changes the answer to ``no'' for one of the relevant in-edges, arbitrarily chosen (in our case the left one).
For the other relevant in-edge $e$, all pointers $p$ with $\last_p = e$ are removed from the set of active pointers, and the corresponding entry in the set $\comp(z)$ is changed from $\undec$ to $\{ 1 \}$.
This also causes three nodes to become inactive.

In iteration $i=4$, the newly born $2$-node performs three merge operations.
In iteration $i=5$, all pointers ending in the root are removed from the set of active pointers due to the existence of a pointer starting in a leaf with the same ``last edge''.
This step causes the set of active pointers to become empty, upon which the algorithm in Phase I terminates.

\paragraph*{Properties of the Phase I Algorithm.}
In the following, we collect some properties of the defined process.
They imply, in particular, that the update rules are well-defined.

\begin{observation}\label{prop:atmostone}
    Each node has, at any point in time, at most one outgoing active pointer.
\end{observation}
\begin{proof}
    This directly follows from the fact that, in the beginning, each node has at most one outgoing active pointer, and in each iteration, each active pointer is either merged into a larger pointer (if it points to a $2$-node), or left unchanged or removed.
\end{proof}

\begin{observation}\label{prop:state}
    If a node is inactive, it can never become active again.
    If a node is an $X$-node, where $X \in \{ 1, 2, 3 \}$, it will never in the further course of the process become a $Y$-node, where $Y\in \{ 1, 2, 3 \}$ and $Y > X$.
\end{observation}
\begin{proof}
    The observation follows from the definitions of the $1$-, $2$-, and $3$-nodes, the definition of $\merge(\cdot,\cdot)$, and the fact that the only new active pointers that are produced during our process are created via $\merge(\cdot,\cdot)$.
\end{proof}

\begin{observation}\label{obs:undec}
    For any $2$-node $v$ with relevant in-edge $e'$ and outgoing edge $e''$, we have $M^e(v) = \undec$ if and only if $e \in \{ e', e'' \}$. 
\end{observation}
\begin{proof}
    Due to the design of the update rules, if $e$ is an outgoing edge (for $v$), then $M^e(v)$ remains $\undec$ indefinitely, and if $e$ is an incoming edge, then $M^e(v)$ is set to some label set $L \neq \undec$ in the first iteration at the end of which there is no active pointer $p$ ending in $v$ and satisfying $\last_p = e$.
    (Note that we use here that there cannot have been a merge of two pointers starting and ending in $v$ so far since otherwise $v$ would not have any incoming pointers and could not be a $2$-node, by \Cref{prop:state}.)
    Since the design of the update rules ensures that once there is no active pointer $p$ ending in $v$ and satisfying $\last_p = e$, this property does not change thereafter, we obtain the lemma statement, by the definition of a relevant in-edge.
\end{proof}

\begin{observation}\label{prop:three}
    At the end of each iteration, it holds that for each active pointer $(u, v)$, the unique path from $u$ to $v$ does not contain a $3$-node, except possibly $u$ and/or $v$.
    Also, if the path from $u$ to $v$ contains a $2$-node $w \neq u$, then the relevant in-edge of $w$ is the edge incoming to $w$ that lies on this path.
\end{observation}
\begin{proof}
    These statements follow since they hold in the beginning of the process and do not change during the process as the merge operation is only ``performed'' by active $2$-nodes (whose relevant in-edge will lie on the path corresponding to the pointer(s) they produce).
    Here, we implicitly use \Cref{prop:state} and the fact that a $2$-node never changes its relevant in-edge (which follows with an analogous argument to the one used in the proof of \Cref{prop:state}).
\end{proof}

\begin{observation}\label{prop:unweird}
    Let $p = (u, v)$ and $p' = (w, x)$ be two active pointers at an end of an iteration, and assume that the unique paths from $u$ to $v$ and from $w$ to $x$ intersect in at least one edge.
    Then there is a directed path that contains all of these four nodes.
\end{observation}
\begin{proof}
    Suppose for a contradiction that this is not the case, which implies that no directed path contains both $u$ and $w$, and let $y$ denote the lowest common ancestor of $u$ and $w$.
    In the beginning, node $y$ is an active $3$-node.
    At the point when $y$ stops being an active $3$-node (which has to happen due to \Cref{prop:three}), the update rules ensure that there is at most one incoming edge $e$ at $y$ such that there exist an active pointer $(z, a)$ such that the unique path from $z$ to $a$ contains $e$.
    Now, the contradiction follows from an analogous argument to the one used in the proof of \Cref{prop:three}.
\end{proof}

\begin{lemma}\label{lem:propcombi}
    At the end of each iteration it holds that (a) for each active node $u$, all nodes on the path from $u$ to $r$ are active, and (b) for any two active nodes $u, v$ such that $v$ is an ancestor of $u$, the node $x$ node $v$ points to is an ancestor of the node $w$ node $u$ points to, or $w = x$.
\end{lemma}
\begin{proof}
    Suppose for a contradiction that the lemma statement is false, and let $i$ be the first iteration such that at the end of iteration $i$ the statement is not satisfied.
    
    Consider first the case that property (b) does not hold, which implies that at the end of iteration $i$, there are active nodes $u, v$ with active pointers $p = (u, w)$, $p' = (v, x)$ such that $v$ is an ancestor of $u$, and $w$ an ancestor of $x$.
    Let $y$ and $z$ denote the nodes $u$ and $v$, respectively, were pointing to at the beginning of iteration $i$.
    Due to the minimality of $i$, we have that $z$ is an ancestor of $y$, or $y = z$, and that $y$ and $z$ are active at the beginning of round $i$.
    The pointer $(u, w)$ must be the result of a merge operation in iteration $i$, as otherwise $w = y$, which would imply that $w$ cannot be an ancestor of $x$.
    Hence, at the beginning of iteration $i$, the active pointer starting at $y$ must be $(y, w)$.
    Also, at the beginning of iteration $i$, we must have $z = x$, or the pointer starting at $z$ must be $(z, x)$ (as otherwise we could not have the active pointer $p' = (v, x)$ at the end of iteration $i$).
    In the latter case, we obtain $y \neq z$ (as otherwise $w = x$), and we see that the nodes $y, z, x, w$ satisfy that, at the beginning of iteration $i$, $z$ is an ancestor of $y$, $w$ is an ancestor of $x$, $z$ points to $x$, and $y$ points to $w$, which yields a contradiction to the minimality of $i$.
    Hence, we can assume that $z = x$.
    This implies that, at the beginning of iteration $i$, $z$ is not an active $2$-node (as otherwise $v$ would not point to $x$ at the end of iteration $i$); since $z$ is an active node (by \Cref{prop:state}) and has an incoming pointer (from $v$), it must be a $3$-node.
    Since $w$ is an ancestor of $x = a$, and we (still) have $y \neq z$, we see that for the active pointer $(y, w)$ at the beginning of iteration $i$, the path from $y$ to $w$ contains a $3$-node that is distinct from both $y$ and $w$, yielding a contradiction to \Cref{prop:three}.
    
    Now consider the second case, namely that property (a) does not hold, which implies that at the end of iteration $i$, there are two vertices $u, v$ such that $v$ is the parent of $u$, $u$ is active, and $v$ is inactive.
    Due to the minimality of $i$ and \Cref{prop:state}, $v$ (as well as $u$) must have been active at the beginning of iteration $i$.
    By the design of the update rules, the only way in which $v$ can have become inactive at the end of iteration $i$ is that the node $w$ node $v$ points to at the beginning of iteration $i$ is a $3$-node that some leaf $x$ satisfying $\last_{(v,w)} = \last_{(x,w)}$ points to as well, at the beginning of iteration $i$.
    Let $y$ denote the end of the active pointer starting in $u$ at the beginning of iteration $i$.
    
    If $u$ lies on the path from $x$ to $w$, then, by property (b), we have that $y$ is an ancestor of $w$, or $y = w$.
    The former cannot be true, as otherwise we would have an active pointer (from $u$ to $y$) at the beginning of iteration $i$ such that the corresponding path contains an internal node that is a $3$-node (namely $w$), which would contradict \Cref{prop:three}.
    However, also the latter cannot be true, as otherwise $u$ would have become inactive at the end of iteration $i$ since $\last_{(u,y)} = \last_{(x,w)}$.
    Hence, $u$ does not lie on the path from $x$ to $w$.
    
    If $y \neq v$, we the two active pointers $(u, y)$ and $(x, w)$ at the beginning of iteration $i$ yield a contradiction to \Cref{prop:unweird}.
    Hence, $y = v$.
    Observe that $i \geq 2$, as at at the beginning of iteration $1$, the only pointers we have are the directed edges of $G'$, and the pointer $(x, w)$ is not such a pointer (as it contains the internal node $v$; we have $x \neq v$ as $x$ is a leaf while $v$ has a child, namely $u$).
    Hence, iteration $i-1$ exists, and at the beginning of iteration $i-1$, node $u$ must have pointed to node $v$ and node $v$ cannot have been a $2$-node, as otherwise we could not have an active pointer $(u,v)$ at the beginning of iteration $i$.
    Since, at the beginning of iteration $i - 1$, node $v$ had an incoming active pointer (from $u$), it cannot have been a $1$-node either, so it must have been a $3$-node.
    Now consider the node $z$ leaf $x$ was pointing to at the beginning of iteration $i-1$.
    As the active pointer starting in $x$ at the beginning of iteration $i$ is $(x,w)$, there are only $3$ possibilities, due to the design of the update rules: 1) $z = w$, or 2) $z \neq v$ lies on the path from $v$ to $w$, or 3) $z \neq v$ lies on the path from $x$ to $v$ and there is an active pointer from $z$ to $w$.
    In either case, we obtain a contradiction to \Cref{prop:three}.
\end{proof}

\begin{lemma}\label{lem:noin23}
    When a node stops being an active $3$-node, it becomes an active $2$-node.
    When a node stops being an active $2$-node, it becomes an active $1$-node.
    When a node stops being an active node, it turns from an active $1$-node into an inactive $1$-node, and remains an inactive $1$-node until the end of Phase I.
    In particular, there are no inactive $2$- or $3$-nodes.
\end{lemma}
\begin{proof}
    Consider an active node $u$ that becomes inactive at the end of some iteration $i$.
    By the update rules, $u$ can only become inactive due to having an active pointer to some node $v \neq u$ at the beginning of iteration $i$, and $v$ having another incoming active pointer from some leaf $w$ (where, potentially, $w = u$) such that $\last_{(u, v)} = \last_{(w, v)}$.
    In particular, $u$ lies on the path from $w$ to $v$.
    At the beginning of iteration $i$, node $u$ cannot be a $3$-node or a $2$-node with the relevant in-edge not lying on the path from $w$ to $v$, since otherwise $u \neq w$, and the active pointer $(w, v)$ together with node $u$ would yield a contradiction to \Cref{prop:three}.
    At the beginning of iteration $i$, node $u$ also cannot be a $2$-node with the relevant in-edge lying on the path from $w$ to $v$ as otherwise $u$ would have an incoming active pointer from some node $x \neq w$ on the path from $w$ to $u$, yielding a contradiction to \Cref{lem:propcombi}.
    Hence, $u$ is a $1$-node at the beginning of iteration $i$.
    By \Cref{prop:state}, the inactive node that $u$ becomes at the end of iteration $i$ must be a $1$-node, and $u$ will remain an inactive $1$-node.
    
    Now consider an active $3$-node that stops being an active $3$-node at the end of some iteration $i$.
    By the above discussion, $u$ is still active at the end of iteration $i$, and, by the design of the update rules, $u$ retains at least one relevant in-edge, which implies that it becomes a $2$-node.
    
    Finally, consider an active $2$-node that stops being an active $2$-node at the end of some iteration $i$.
    Again, we obtain that $u$ is still active at the end of iteration $i$, and, again by the design of the update rules, we see that there can be at most one edge $e$ incoming at $u$ such that there exists an active pointer $p$ satisfying $\last_p = e$, which implies that $u$ is a $1$-node at the end of iteration $i$ (as $u$ stops being an active $2$-node).
\end{proof}

Due to \Cref{lem:noin23}, we will not have to distinguish between active and inactive $2$-nodes (or $3$-nodes) in the remainder of the paper as we know that such nodes cannot be inactive.

\paragraph*{Bounding the Number of Iterations.}

In the following, we will fix some notation that is required to prove that the number of iterations until Phase I terminates is in $O(\log n)$.
We will denote the induced tree consisting of \emph{active} nodes at the end of iteration $i$ by $T_i$; we set $T_0 \coloneqq G'$.
Due to \Cref{lem:propcombi}, we know that $T_i$ is indeed a (rooted) subtree of $G'$, and that its root is the root of $G'$, namely $\rooot$; we also know that, for any $i \geq 1$, tree $T_i$ is an induced subtree of $T_{i-1}$, due to \Cref{prop:state}.
For each $T_i$, we denote the maximal connected components consisting of non-root degree-$2$ nodes by $B_{i,1}, \dots, B_{i,z_i}$, in an arbitrary, but fixed, order.
Here $z_i$ denotes the number of such maximal connected components in $T_i$.
We call the $B_{i,j}$ \emph{blocks} of $T_i$.
For simplicity, we will also use $B_{i,j}$ to denote the set of nodes of $B_{i,j}$.
In the following we will collect some insights about the $T_i$ and $B_{i,j}$.

\begin{lemma}\label{lem:leaves}
    Consider some iteration $i \geq 1$.
    Any leaf $u \neq \rooot$ of $T_i$ is also a leaf of $T_{i-1}$.
    Moreover, any leaf $u \neq \rooot$ of $T_i$ is also a leaf of $G'$.
\end{lemma}
\begin{proof}
    For a contradiction, suppose that, for some $i \geq 1$, tree $T_i$ contains a leaf $u \neq \rooot$ that is not a leaf of $T_{i-1}$.
    Note that $u$ cannot have any incoming active pointer at the end of iteration $i$, and therefore must be a $1$-node at that point in time.
    Let $v$ be a child of $u$ in $T_{i-1}$, and let $w$ denote the node $v$ is pointing to at the end of iteration $i-1$.
    Due to our assumption, $v$ is active at the end of iteration $i-1$, but inactive at the end of iteration $i$.
    Due to the design of our update rules, the only way in which this can happen is that at the end of iteration $i-1$, $w$ is a $3$-node or the root, and there is an active pointer $(x, w)$ from some leaf $x$ with $\last_{(x, w)} = \last_{(v, w)}$.
    Observe that, by \Cref{lem:noin23}, $u$ cannot be a $3$-node at the end of iteration $i - 1$ (as it is a $1$-node at the end of iteration $i$), which implies $w \neq u$.
    Hence, $w$ is an ancestor of $u$, and, by \Cref{prop:three} and \Cref{lem:propcombi}, it follows that at the end of iteration $i-1$, the active pointer starting at $u$ must end in $w$, and $\last_{(u, w)} = \last_{(x, w)}$.
    But this implies, again by the design of the update rules, that if $v$ becomes inactive at the end of iteration $i$, then so does $u$.
    This yields a contradiction to the fact that $u$ is active at the end of iteration $i$.
    
    Since we showed that any leaf $u \neq \rooot$ in the tree of active nodes at the end of some iteration is also a leaf in the tree of active nodes at the end of the previous iteration, we obtain, by applying this argumentation iteratively, that $u$ must also be a leaf in $T_0 = G'$.
\end{proof}

\begin{corollary}\label{cor:niceblocks}
    Consider some iteration $i \geq 1$, and some block $B_{i,j}$.
    Let $u$ be a node in $B_{i,j}$.
    If $u$ is in some block $B_{i-1,j'}$, then $B_{i-1,j'} \subseteq B_{i,j}$ (considered as node sets).
\end{corollary}
\begin{proof}
    Let $u$ be as described in the lemma, and suppose, for a contradiction, that there is some node $v \neq u$ satisfying $v \in B_{i-1,j'}$ and $v \notin B_{i,j}$.
    By the definition of blocks, either $v$ is an ancestor of $u$, or $u$ is an ancestor of $v$.
    In the former case, observe that, due to \Cref{lem:propcombi}, all ancestors of $u$ in $T_{i-1}$ are also contained in $T_i$, which implies that $v$ is a degree-$2$ node in $T_i$ belonging to $B_{i,j}$, yielding a contradiction.
    In the latter case, observe that $v$ and its child in $T_{i-1}$ must be contained in $T_i$ as otherwise some node on the path from $v$ to $u$ must be a leaf in $T_i$ while not being a leaf in $T_{i-1}$, which would contradict \Cref{lem:leaves}.
    Now we obtain a contradiction in an analogous way to the previous case.
\end{proof}

For each block $B_{i,j}$ with $i \geq 1$, we denote the set of blocks $B_{i-1,j'}$ that have non-empty intersection with $B_{i,j}$ by $\prev(B_{i,j})$.
Due to \Cref{cor:niceblocks}, we know that the union of all node sets contained in $\prev(B_{i,j})$ is a subset of $B_{i,j}$.
Moreover, for each block $B_{i,j}$ with $i \geq 1$, we denote the set of vertices in $B_{i,j}$ that are not contained in some $B \in \prev(B_{i,j})$ by $\new(B_{i,j})$.

We will also need the notion of a pointer chain.
\begin{definition}[pointer chain]
    A \emph{pointer chain} (from a node $c_0$ to a node $c_y$) at the end of some iteration $i$ is a finite sequence $C = (c_0, \dots, c_y)$ of nodes such that for any $1 \leq j \leq y$, there is an active pointer $(c_{j-1}, c_j)$ at the end of iteration $i$.
    We call a pointer chain a \emph{leaf-root pointer chain} if $c_0$ is a leaf in $G'$ and $c_y = \rooot$.
\end{definition}
Note that any pointer chain at the end of some iteration $i$ consists only of active nodes, i.e., of nodes from $T_i$ (this holds for the last node in the pointer chain due to \Cref{lem:propcombi}).

\begin{observation}
    For any iteration $i$, and any leaf $u \neq \rooot$ in $T_i$, there exists a leaf-root pointer chain from $u$ to $\rooot$ at the end of iteration $i$.
\end{observation}
\begin{proof}
    This follows directly from the definition of a pointer and \Cref{lem:propcombi,lem:leaves}.
\end{proof}

In order to maintain a certain guarantee (given in \Cref{lem:leafroot}) throughout Phase I (that will help us to bound the number of iterations), we will need to assign an integer value $k_{i,j}$ to each block $B_{i,j}$ that, roughly speaking, provides an upper bound for the number of nodes from $B_{i,j}$ contained in any leaf-root pointer chain.
Define $\prev(k_{i,j})$ to be the set of all indices $j'$ such that $B_{i-1,j'} \in \prev(B_{i,j})$.
For each block $B_{0,j}$, we set $k_{0,j} \coloneqq |B_{0,j}|$.
For each block $B_{i,j}$ with $i \geq 1$, we set
\[
  k_{i,j} \coloneqq |\new(B_{i,j})| + 1/2 \cdot \sum_{j' \in \prev(k_{i,j})} k_{i-1, j'}\enspace.
\]

\begin{lemma}\label{lem:leafroot}
    Consider a leaf-root pointer chain $C = (c_0, \dots, c_y = \rooot)$ at the end of some iteration $i \geq 0$ (where we set the end of iteration $0$ to be the starting point of our process).
    For each block $B_{i,j}$, the number of nodes contained in $C \cap B_{i,j}$ is at most $k_{i,j}$.
\end{lemma}
\begin{proof}
    We prove the statement by induction in $i$.
    For $i = 0$, the statement trivially holds, by the definition of $k_{0,j}$.
    Now, consider some $i \geq 1$, and assume that the statement holds for $i-1$.
    Let $C = (c_0, \dots, c_y)$ be an arbitrary leaf-root pointer chain at the end of iteration $i$, and let $C' = (c'_0 = c_0, c'_1, \dots, c'_{y'})$ denote the leaf-root pointer chain starting at $c_0$ at the end of iteration $i-1$.
    By the design of the update rules, the definition of the function $\merge(\cdot, \cdot)$, and \Cref{lem:propcombi}, the sequence $C$ is a subsequence of $C'$, i.e., $C$ is obtained from $C'$ by removing elements.
    Furthermore, we observe that any (non-root) degree-$2$ node in $T_{i-1}$ with an incoming active pointer at the end of iteration $i-1$ must be a $2$-node (by the definition of a $2$-node), and any (non-root) node $u$ of degree at least $3$ in $T_{i-1}$ must be a $3$-node at the end of iteration $i-1$ (as, for each child $v$ of $u$ in $T_{i-1}$, there must be an active pointer $(v, u)$, due to \Cref{prop:three} and \Cref{prop:unweird}).
    
    Consider an arbitrary block $B_{i,j}$, and an arbitrary block $B_{i-1,j'} \in \prev(B_{i,j})$.
    Due to the definitions of a pointer and a block, the nodes in $C' \cap B_{i-1,j'}$ form a subsequence of $C'$ consisting of consecutive nodes $c'_p, \dots, c'_q$.
    By the definition of $\merge(\cdot, \cdot)$ and the fact that the nodes in $C'' \coloneqq (c'_p, \dots, c'_q)$ are degree-$2$ nodes in $T_{i-1}$ (and hence $2$-nodes at the end of iteration $T_{i-1}$), we see that for any two consecutive nodes in $C''$, at most one of the nodes is contained in $C$ (by the design of the update rules).
    Moreover, as $\merge$ operations are only ``performed'' by $2$-nodes, any node in $C'$ that is a (non-root) node of degree at least $3$ in $T_{i-1}$ (and hence a $3$-node) will be contained in $C$, which implies that $c'_p$ is not contained in $C$ (as $c'_{p-1}$ is contained in $C$ and points to $c'_{p+1}$ at the end of iteration $i$).
    Hence, we can conclude that for each $B_{i-1,j'}$, we have $|C \cap B_{i-1,j'}| \leq 1/2 \cdot |C' \cap B_{i-1,j'}| \leq 1/2 \cdot k_{i-1,j'}$.
    
    By \Cref{cor:niceblocks}, we have
    \[
        B_{i,j} = \new(B_{i,j}) \cup \bigcup_{B \in \prev(B_{i,j})} B \enspace,
    \]
    which yields
    \[
        \left|C \cap B_{i,j}\right| \leq \left|\new(B_{i,j})\right| + \sum_{B \in \prev(B_{i,j})} \left|C \cap B\right| \leq \left|\new(B_{i,j})\right| + \sum_{j' \in \prev(k_{i,j})} \left(1/2 \cdot k_{i-1,j'}\right) = k_{i,j}
    \]
    as desired.
\end{proof}

In order to bound the number of iterations in Phase I, we will make use of a potential function argument.
Recall that $z_i$ denotes the number of blocks of $T_i$.
For each $i \geq 0$ (such that Phase I has not terminated after $i - 1$ iterations), set $\Phi_i \coloneqq \Phi'_i + \Phi''_i$, where $\Phi'_i$ is the number of leaves in $T_i$, and $\Phi''_i \coloneqq \sum_{1 \leq j \leq z_i} k_{i,j}$.

\begin{lemma}\label{lem:calc}
    Consider any iteration $i \geq 1$ such that Phase I does not terminate at the end of iteration $i$ or $i+1$.
    Then $\Phi_{i+1} \leq 7/8 \cdot \Phi_{i-1}$.
\end{lemma}
\begin{proof}
    For $i' \in \{ i-1, i \}$, let $X_j$ denote the number of leaves that are contained in $T_{i'}$, but not in $T_{i'+1}$.
    By \Cref{lem:leaves}, $X_{i'} = \Phi'_{i'} - \Phi'_{i'+1}$.
    
    By the definition of the $k_{a,j}$ and $\new(B_{a,j})$ (as well as \Cref{cor:niceblocks}), we have $\Phi''_{i'+1} \leq 1/2 \cdot \Phi''_{i'} + |\new(i'+1)|$, where $\new(i'+1)$ denotes the set of (non-root) nodes in $T_{i'+1}$ that have degree $2$ in $T_{i'+1}$ but not in $T_{i'}$.
    Recall (from the proof of \Cref{lem:leafroot}) that any (non-root) node of degree at least $3$ in $T_{i'}$ must be a $3$-node at the end of iteration $i'$, and observe that any leaf in $T_{i'}$ must be a $1$-node at the end of iteration $i'$.
    By the design of the update rules, it follows that to any node $u$ from the set $\new(i'+1)$, we can assign a leaf $f_u$ of $T_{i'}$ such that there is an active pointer $(f_u,u)$ at the end of iteration $i'$, and $f_u$ becomes inactive at the end of iteration $i'+1$.
    As $f_u \neq f_{u'}$ for any two nodes $u \neq u'$ from $\new(i'+1)$ (due to \Cref{prop:atmostone}), we obtain $|\new(i'+1)| \leq X_{i'}$, which implies
    \[
        \Phi''_{i'+1} \leq 1/2 \cdot \Phi''_{i'} + X_{i'} \enspace.
    \]
    As the next step, we bound $\Phi'_{i'+1}$ in terms of $\Phi''_{i'}$.
    Let $\stays(i')$ be the set of all leafs $u$ of $T_{i'}$ such that at the end of iteration $i'$ the active pointer starting in $u$ does not end in a $3$-node or the root.
    Since all (non-root) nodes that have degree at least $3$ in $T_{i'}$ are $3$-nodes at the end of iteration $i'$ (as already observed above), any node $u \in \stays(i')$ must point to some (non-root) degree-$2$ node $f_u$ in $T_{i'}$.
    For any two distinct nodes $u, u' \in \stays{i'}$, the nodes $f_u$ and $f_{u'}$ must lie in different blocks of $T_{i'}$ (due to \Cref{prop:three}), and each block $B_{i',j}$ containing such a node $f_u$ must satisfy $k_{i',j} \geq 1$ (as the leaf-root pointer chain starting in $u$ contains at least one node of $B_{i',j}$, namely $f_u$).
    Hence, $|\stays(i')| \leq \Phi''_{i'}$.
    Moreover, the design of the update rules ensures that out of all the leaves in $T_{i'}$ that point to a $3$-node or the root at the end of iteration $i'$, at least half will become inactive at the end of iteration $i'+1$.
    Thus, we obtain
    \[
       \Phi'_{i'+1} \leq |\stays(i')| + 1/2 \cdot \left(\Phi'_{i'} - |\stays(i')|\right) = 1/2 \cdot \left( \Phi'_{i'} + |\stays(i')| \right) \leq 1/2 \cdot \left( \Phi'_{i'} + \Phi''_{i'} \right) \enspace.
    \]
    
    To finish our calculations, we consider two cases.
    Let us first consider the case that $\Phi''_{i-1} \geq 1/3 \cdot \Phi'_{i-1}$, which implies $\Phi'_{i-1} \leq 3/4 \cdot \Phi_{i-1}$.
    Then, using the equations and inequalities derived above, we obtain
    \begin{align*}
        \Phi_i = \Phi'_i + \Phi''_i &\leq \left( \Phi'_{i-1} - X_{i-1} \right) + \left( 1/2 \cdot \Phi''_{i-1} + X_{i-1} \right)\\
        &= \Phi'_{i-1} + 1/2 \cdot \Phi''_{i-1} = 1/2 \cdot \Phi_{i-1} + 1/2 \cdot \Phi'_{i-1} \leq 7/8 \cdot \Phi_{i-1} \enspace.
    \end{align*}
    Similarly to above, we see that
    \[
        \Phi_{i+1} \leq \Phi'_i + 1/2 \cdot \Phi''_i \leq \Phi_i \enspace,
    \]
    which implies
    \[
        \Phi_{i+1} \leq 7/8 \cdot \Phi_{i-1} \enspace.
    \]
    
    Now, consider the case that $\Phi''_{i-1} < 1/3 \cdot \Phi'_{i-1}$.
    Using the inequalities derived earlier, we obtain
    \[
        \Phi'_i \leq 1/2 \cdot \left( \Phi'_{i-1} + \Phi''_{i-1} \right) \leq 2/3 \cdot \Phi'_{i-1}
    \]
    and
    \[
        \Phi''_i \leq 1/2 \cdot \Phi''_{i-1} + X_{i-1} \leq \Phi''_{i-1} + \Phi'_{i-1} - \Phi'_i \enspace.
    \]
    Similarly to the previous case, we see that
    \begin{align*}
        \Phi_{i+1} \leq \Phi'_i + 1/2 \cdot \Phi''_i \leq &\Phi'_i + 1/2 \cdot \left( \Phi''_{i-1} + \Phi'_{i-1} - \Phi'_i \right)\\
        &= 1/2 \cdot \Phi''_{i-1} + 1/2 \cdot \left( \Phi'_{i-1} + \Phi'_i \right)\\
        &\leq 1/2 \cdot \Phi''_{i-1} + 5/6 \cdot \Phi'_{i-1} \leq 5/6 \cdot \Phi_{i-1} \enspace.
    \end{align*}
    Hence, in both cases, we have $\Phi_{i+1} \leq 7/8 \cdot \Phi_{i-1}$, as desired.
\end{proof}

Using \Cref{lem:calc}, we are finally able to bound the number of iterations in Phase I.

\begin{lemma}\label{lem:runtimephaseone}
    Phase I terminates after $O(\log n)$ iterations.
\end{lemma}
\begin{proof}
    Suppose for a contradiction that there is no constant $c$ such that Phase I always terminates after at most $c \cdot \log n$ iterations.
    Observe that $\Phi_0 = \Phi'_0 + \Phi''_0 \leq n + n = 2n$.\footnote{Note that the compatibility tree $G'$ has actually only $O(n/\log n)$ nodes, by \Cref{cor:fewernodes}, but upper bounding this by $n$ suffices.}
    By \Cref{lem:calc}, there exists some constant $c$ such that $\Phi_{c \cdot \log n} < 1$.
    By the definition of $\Phi_i$, it follows that the tree $T_{c \cdot \log n}$ of active nodes obtained after $c \cdot \log n$ iterations does not contain any leaves (apart from, potentially, the root).
    This implies that there is no active pointer after $c \cdot \log n$ iterations, which implies that Phase $1$ terminates after at most $c \cdot \log n$ iterations, yielding a contradiction.
\end{proof}

\paragraph{Phase II}\label{sec:phasetwo}

Let $\ps_{\fin}$ denote the set of pointers at the end of the last iteration of Phase I.
In Phase II, we will go through the pointers of some subset of $\ps_{\fin}$ in some order and ``fix'' them, i.e., for each such pointer $p = (u, v)$ we assign to the two half-edges $(u, \first_p)$ and $(v, \last_p)$ a label from $\sout$ each.
In order to describe the order in which we process the pointers, we group the pointers we want to process into sets $\timey(1), \timey(2), \dots$.
We will process each of the pointers in set $\timey(i)$ in parallel in iteration $i$.

Define $\timey(1)$ to be the set of all pointers $p = (u, \rooot) \in \ps_{\fin}$ for which $u$ is a leaf.
For each $i \geq 2$, define $\timey(i)$ to be the set of all pointers $p'$ such that there exists a pointer $p = (u, v) \in \timey(i-1)$ satisfying
\begin{enumerate}
    \item $p' = (u, \pred_p)$,
    \item $p' = (\pred_p, v)$, or
    \item $p' = (w, \pred_p)$ where $w$ is a leaf and the edge $\last_{p'}$ does not lie on the path from $u$ to $v$.
\end{enumerate}

In the following, we collect some insights about the pointers in $\timey(i)$.
We start with a helper lemma that provides information about the leaf-root pointers produced in Phase I.

\begin{lemma}\label{lem:attheroot}
    For each edge $e$ incoming to the root $\rooot$, there is precisely one pointer $p = (u, \rooot) \in \ps_{\fin}$ such that $u$ is a leaf and $\last_p = e$.
\end{lemma}
\begin{proof}
    Fix an arbitrary edge $e$ incoming to the root $\rooot$.
    In the beginning of Phase I, there is an active pointer $p'$ that ends in $\rooot$ and satisfies $\last_{p'} = e$.
    Due to the design of the update rules for Phase I, this can only change once such a pointer that additionally starts in a leaf has been added to the pointer set.
    As at the end of Phase I, there is no active pointer left, it follows that there is at least one pointer $p$ that ends in $\rooot$ and satisfies $\last_{p} = e$.
    
    In order to show that there is at most such pointer, consider the first iteration $i$ in which such a pointer appeared in the set of active pointers (and therefore also in the set of pointers).
    By \Cref{prop:unweird}, there can only be one such pointer at the end of iteration $i$.
    Moreover, due to the design of the update rules, no pointer $p''$ satisfying $\last_{p''} = e$ is added to the set of pointers in any later iteration (as no active pointer $p'''$ satisfying $\last_{p'''} = e$ remains at the end of iteration $i+1$).
    It follows that for edge $e$, there is precisely one pointer as described in the lemma.
\end{proof}

We continue by proving a lemma that highlights which pointers in $\timey(i)$ are ``produced'' by some pointer in $\timey(i-1)$.

\begin{lemma}\label{lem:nicesplit}
    Let $p = (u,v)$ be a pointer in $\ps_{\fin}$ with $\pred_p \neq \bot$.
    Then $\pred_p$ has degree at least $2$ in $G'$.
    Moreover,
    \begin{enumerate}
        \item if $\pred_p$ has degree $2$, then $(u, \pred_p), (\pred_p, v) \in \ps_{\fin}$, and
        \item if $\pred_p$ has degree at least $3$, then $(u, \pred_p), (\pred_p, v) \in \ps_{\fin}$, and for each edge $e$ incoming at $\pred_p$ that does not lie on the path from $u$ to $v$, there is exactly one pointer $p' = (w, \pred_p) \in \ps_{\fin}$ such that $w$ is a leaf and $\last_{p'} = e$.
    \end{enumerate}
\end{lemma}
\begin{proof}
    Since $\pred_p \neq \bot$, the pointer $p$ must be the result of a merge operation, which, by the definition of $\merge(\cdot,\cdot)$, implies that $\pred_p$ is an internal vertex of the path from $u$ to $v$.
    Hence, $\deg(\pred_p) \geq 2$.
    
    First, consider the case that $\deg(\pred_p) = 2$.
    From the design of the update rules of Phase I and the definition of $\merge(\cdot,\cdot)$, it follows directly that at some point during Phase I, there must have existed active pointers $(u, \pred_p), (\pred_p, v)$.
    This implies $(u, \pred_p), (\pred_p, v) \in \ps_{\fin}$.
    
    Now, consider the case that $\deg(\pred_p) \geq 3$.
    Analogously to the previous case, we obtain $(u, \pred_p), (\pred_p, v) \in \ps_{\fin}$.
    Now what is left to be shown follows from an analogous argumentation to the one provided in the proof of \Cref{lem:attheroot}, with only one difference: for the considered edge $e$ incoming at $\pred_p$, it could also be the case that, at the point in Phase I where the property that there is an active pointer $p''$ that ends in $\pred_p$ and satisfies $\last_{p''} = e$ becomes false, this happens due to step~\ref{step2}, and not due to step~\ref{step3}, of the update rules.
    However, observe that any vertex $x$ can ``perform'' merge operations in at most one iteration during Phase I (where the node considered to perform a merge operation is the node $\pred_{p'''}$ where $p'''$ is the pointer created during the merge operation) since, by the design of the update rules, the merge operations of that iteration will make sure that no active pointer that ends in $x$ remains (which cannot change thereafter).
    Observe further that for $x = \pred_p$, each of those merge operations must have merged two pointers where the one incoming at $\pred_p$ (let us call it $q$) satisfies $\last_q = e'$ where $e'$ is the edge incoming at $\pred_p$ that lies on the path from $u$ to $v$.
    Hence, the aforementioned difference only applies to pointers $q$ ending at $\pred_p$ satisfying $\last_q = e'$, and since the lemma statement only concerns pointers $p'$ with $\last_{p'} \neq e'$, that difference is irrelevant, and we can simply apply the argumentation from the proof of \Cref{lem:attheroot}.
\end{proof}

    

The next lemma shows that the sets $\timey(i)$ yield a partition of the edge set in a natural way.
For this result we need to introduce a bit of notation.
We call a pointer $(u,v)$ such that $(u,v)$ is an edge of $G'$ a \emph{basic} pointer.
Moreover, we denote by $\done(i)$ the set of all basic pointers contained in $\timey(1) \cup \dots \cup \timey(i)$.
For simplicity, also define $\done(0) \coloneqq \emptyset$.
Finally, for any two nodes $u, v$ such that $v$ is an ancestor of $u$, we denote by $\betw(u,v)$ the set of all edges $(w,x)$ such that 1) $(w,x) = (y,v)$, or 2) $y$ is an ancestor of $w$, but $u$ is not an ancestor of $w$, where $y$ is the child of $v$ that lies on the path from $u$ to $v$.
In other words, $\betw(u,v)$ is the set of all edges that can be reached both from $u$ without crossing $v$, and from $v$ without crossing $u$.
For simplicity, for any pointer $p = (u,v)$, we also define $\betw(p) \coloneqq \betw(u,v)$.

\begin{lemma}\label{lem:between}
    Consider any $i \geq 1$, and any edge $e = (u,v) \in E(G')$.
    If $\done(i-1)$ does not contain the pointer $p = (u,v)$, then there is exactly one pointer $(w,x) \in \timey(i)$ such that $e \in \betw(w,x)$.
    If $\done(i-1)$ contains the pointer $p = (u,v)$, then there is no pointer $(w,x) \in \timey(i)$ such that $e \in \betw(w,x)$.
\end{lemma}
\begin{proof}
    We prove the statement by induction in $i$.
    For $i = 1$, we have $\done(i-1) = \emptyset$, so $p \notin \done(i-1)$.
    By \Cref{lem:attheroot}, for each edge $e'$ incident to $\rooot$, there is exactly one pointer $p_{e'} \in \timey(1)$ with $\last_{p_{e'}} = e'$, and $p_{e'}$ is guaranteed to be a leaf-root pointer.
    By the definition of $\betw(\cdot)$, it follows that $e$ is contained in $\betw(p_{e'})$ where $e'$ is the unique edge incident to $\rooot$ that lies on the path connecting $e$ with $\rooot$, and that $e$ is not contained in $\betw(p'')$ for any pointer $p'' \neq p_{e'}$ from $\timey(1)$.
    This covers the base of the induction.
    
    For the induction step assume that the lemma statement holds for $i-1$ (where $i \geq 2$); we aim to show that it then also holds for $i$.
    Consider first the case that $\done(i-1)$ does not contain the pointer $p = (u,v)$.
    Then also $\done(i-2)$ does not contain the pointer $p = (u,v)$, and the induction hypothesis guarantees that there is exactly one pointer $p' = (y,z) \in \timey(i-1)$ such that $e \in \betw(p')$.
    By the definitions of $\timey(i)$ and $\betw(\cdot)$, the only pointers $p'' \in \timey(i)$ that could possibly satisfy $e \in \betw(p'')$ are $(y, \pred_{p'})$, $(\pred_{p'}, z)$, and some $(a, \pred_{p'})$ where $a$ is a leaf and $\last_{(a, \pred_{p'})}$ does not lie on the path from $y$ to $z$.
    Now, \Cref{lem:nicesplit} (together with the definition of $\timey(i)$) guarantees that $e$ is contained in $\betw(p'')$ for exactly one of those possible choices for $p''$, since the sets $\betw(y, \pred_{p'})$, $\betw(\pred_{p'}, z)$, $\betw(b_1, \pred_{p'}), \dots, \betw(b_{\deg(\pred_{p'})-2}, \pred_{p'})$ are pairwise disjoint and their union is $\betw(y, z)$.
    (Here, $b_1, \dots, b_{\deg(\pred_{p'})-2}$ are the starting vertices of the precisely $\deg(\pred_{p'})-2$ pointers $p''' \in \timey(i)$ ending in $\pred_{p'}$ and satisfying that $\last_{p'''}$ does not lie on the path from $y$ to $z$, whose existence is guaranteed by \Cref{lem:nicesplit}.)
    Hence, there is exactly one pointer $(w,x) \in \timey(i)$ such that $e \in \betw(w,x)$, as desired.
    
Now consider the second case, i.e., that $\done(i-1)$ contains the pointer $p = (u,v)$.
    By the induction hypothesis, there is no pointer $p' \neq p$ contained in $\timey(i-1)$ such that $e \in \betw(p')$.
    By the definitions of $\timey(i)$ and $\betw(\cdot)$ (and the fact that $\pred_p = \bot$), it follows that there is no pointer $(w,x) \in \timey(i)$ such that $e \in \betw(w,x)$, as desired.
\end{proof}

For any $i \geq 1$, and any pointer $p = (u,v) \in \timey(i)$, define $\succc(p)$ to be the set of all pointers $p' \in \timey(i+1)$ satisfying
\begin{enumerate}
    \item $p' = (u, \pred_p)$,
    \item $p' = (\pred_p, v)$, or
    \item $p' = (w, \pred_p)$ where $w$ is a leaf and the edge $\last_{p'}$ does not lie on the path from $u$ to $v$.
\end{enumerate}
If $\pred_p = \bot$, set $\succc(p) \coloneqq \emptyset$.
We obtain the following observation.

\begin{observation}\label{obs:summary}
    For any $i \geq 2$, and any pointer $p' \in \timey(i)$, there is exactly one pointer $p \in \timey(i-1)$ such that $p' \in \succc(p)$.
    Moreover, for any $i \geq 1$, and any pointer $p = (u,v) \in \timey(1)$ with $\pred_p \neq \bot$, we have $\succc(p) = \{ (u, \pred_p), (\pred_p, v), p_1, \dots, p_{\deg(\pred_p)-2} \}$ where each $p_j$ is a pointer starting in a leaf, ending in $\pred_p$, and satisfying $\last_{p_j} = e_j$, where $e_1, \dots, e_{\deg(\pred_p)-2}$ are the $\deg(\pred_p)-2$ edges incoming to $\pred_p$ that do not lie on the path from $u$ to $v$. 
\end{observation}
\begin{proof}
    The observation follows from the definition of $\timey(\cdot)$ and \Cref{lem:nicesplit,lem:between}.
\end{proof}

\paragraph*{The Algorithm Description.}

Now we are set to describe how our algorithm $\atwo$ proceeds in Phase II.
After providing the description, we will prove that the algorithm is well-defined and correct.

The algorithm proceeds in iterations $i = 1, 2, \dots$ where in each iteration $i$, we process all the pointers contained in $\timey(i)$.
When processing a pointer $p = (u,v)$, we assign some output label from $\sout$ to each so-far-unlabeled half-edge from $\{ (u, \first_p), (\last_p, v) \}$.
Due to \Cref{obs:summary}, it suffices to explain
\begin{itemize}
    \item[(a)]\label{point1} how we choose those output labels for each pointer in $\timey(1)$, and
    \item[(b)]\label{point2} for each already processed pointer $p$ with $\pred_p \neq \bot$, how we choose those output labels for each pointer in $\succc(p)$.
\end{itemize}
For point (a), let $p_1 = (u_1, \rooot), \dots, p_k = (u_k, \rooot)$ denote the pointers in $\timey(1)$.
By \Cref{lem:attheroot}, we know that $k = \deg(\rooot)$ and, for each edge $e$ incident to $\rooot$, there is precisely one pointer $p_j$ with $\last_{p_j} = e$.
Recall \Cref{def:comptree,def:merge}.
We first assign labels to the half-edges incident to $\rooot$.
More precisely, for each edge $e \in \incid(\rooot)$, assign to half-edge $(\rooot, e)$ some label $\gout((\rooot, e)) \coloneqq \ell^e$ such that, for the obtained tuple $(\ell^e)_{e \in \incid(\rooot)}$, we have
\begin{enumerate}
    \item $(\ell^e)_{e \in \incid(\rooot)} \in S_{\rooot}$, and
    \item $\ell^e \in \comp^e(\rooot)$, for each $e \in \incid(\rooot)$.
\end{enumerate}
Now, for each pointer $p_j$, we assign to half-edge $(u_j, \first_{p_j})$ some label $\gout((u_j, \first_{p_j})) \coloneqq \ell^*$ such that $(\ell^*, \ell^{\last_{p_j}}) \in \pairs_{p_j}$ and $(\ell^*) \in S_{u_j}$.

For point (b), let $p = (u,v)$ denote an already processed pointer with $\pred_p \neq \bot$.
Note that, for the two pointers $p' \coloneqq (u, \pred_p)$ and $p'' \coloneqq (\pred_p, v)$, the half-edges $(u, \first_{p'})$ and $(\last_{p''}, v)$ have already been assigned output labels since $p$ has already been processed; denote those output labels by $\ell$ and $\ell'$, respectively.
However, by \Cref{lem:between,obs:summary}, these are the only half-edges that are already labeled, out of all the half-edges that ``by definition'' have to be labeled after processing the pointers in $\succc(p)$.
Out of these unlabeled half-edges, we first assign an output to all half-edges incident to $\pred_p$.
Concretely, for each edge $e \in \incid(\pred_p)$, assign to half-edge $(\pred_p, e)$ some label $\gout((\pred_p, e)) \coloneqq \ell^e$ such that, for the obtained tuple $(\ell^e)_{e \in \incid(\pred_p)}$, we have
\begin{enumerate}
    \item $(\ell^e)_{e \in \incid(\pred_p)} \in S_{\pred_p}$,
    \item $\ell^e \in \comp^e(\pred_p)$, for each $e \in \incid(\pred_p) \setminus \{\last_{p'}, \first_{p''}\}$, and
    \item $(\ell, \ell^{\last_{p'}}) \in \pairs_{p'}$ and $(\ell^{\first_{p''}}, \ell') \in \pairs_{p''}$.
\end{enumerate}
Finally, for each pointer $p''' = (w, \pred_p)$ where $w$ is a leaf and $\last_{p'''}$ does not lie on the path from $u$ to $v$, we assign to half-edge $(w, \first_{p'''})$ some label $\gout((w, \first_{p'''})) \coloneqq \ell^*$ such that $(\ell^*, \ell^{\last_{p'''}}) \in \pairs_{p'''}$ and $(\ell^*) \in S_{w}$.
By \Cref{obs:summary}, this finishes the processing of all the pointers in $\succc(p)$.

The algorithm in Phase II terminates in the first iteration $i$ in which $\timey(i) = \emptyset$.
This concludes the description of $\atwo$.

\paragraph*{Well-Definedness and Correctness.}

From the description of the algorithm in Phase II it is not clear that the labels with certain properties the algorithm is supposed to output do actually exist.
In order to show that the algorithm is well-defined, we first need a helper lemma based on the following definitions.

\begin{definition}
    Let $e = (u, v)$ be an edge in $G'$.
    We denote the set of vertices that have ancestor $u$ or are equal to $u$ by $V_e$, the set of edges with at least one endpoint in $V_e$ by $E_e$, and the set of half-edges $(w, e')$ with $e' \in E_e$ by $H_e$.
    A labeling of the half-edges in $H_e$ (with a label from $\sout$ each) is a \emph{correct solution} for $H_e$ if it is a correct solution on the compatibility subtree $CS_e$ induced by $V_e \cup \{ u \}$ where we do not have any constraint for node $u$ (in the definition of a correct solution for a compatibility tree (see \Cref{def:comptree})).
\end{definition}

\begin{definition}
    Let $u,v$ be two nodes such that $v$ is an ancestor of $u$.
    We denote the set of nodes that are an endpoint of some edge in $\betw(u,v)$ by $V_{u,v}$, and the set of half-edges $(w, e)$ with $e \in \betw(u,v)$ by $H_{u,v}$.
    A labeling of the half-edges in $H_{u,v}$ (with a label from $\sout$ each) is a \emph{correct solution} for $H_{u,v}$ if it is a correct solution on the compatibility subtree $CS_{u,v}$ induced by $V_{u,v}$ where we do not have any constraint for nodes $u$ and $v$ (in the definition of a correct solution for a compatibility tree).
\end{definition}

\begin{lemma}\label{lem:encodecomplete}
    At the beginning of Phase I (i.e., for $i = 0$), and after any iteration $i$, the following two properties hold.
    \begin{enumerate}
        \item For any half-edge $(u,e)$ satisfying $M^e(u) \neq \undec$, the set $M^e(u)$ contains precisely the output labels $\ell$ such that there exists a labeling of the half-edges in $H_e$ that is a correct solution on $CS_e$ and labels $(u,e)$ with $\ell$.
        \item For any pointer $p = (v,w) \in \ps_i$, the set $\pairs_p$ contains precisely the pairs $(\ell, \ell')$ of output labels such that there exists a labeling of the half-edges in $H_{v,w}$ that is a correct solution on $CS_{v,w}$ and labels $(v,\first_p)$ with $\ell$ and $(\last_p,w)$ with $\ell'$.
    \end{enumerate}
\end{lemma}
\begin{proof}
    We prove the statement by induction in the first iteration $i$ in which $M^e(u)$ was set to some label set $L \neq \undec$, resp.\ in which $p$ was added to the set of pointers.
    The base of the induction is implied by the initialization of $M^e(u)$ and the pointer set at time $i = 0$.
    The induction step directly follows from the precise definition of the two steps in the update rules that create new pointers and change the values $M^e(u)$, namely step~\ref{step2} (which relies on the precise definition of the merge operation) and step~\ref{itemb}, respectively.
\end{proof}

Now, we are set to show that the algorithm for Phase II is well-defined and correct.

\begin{lemma}\label{lem:wellandcorrect}
    Assuming that a correct solution for the compatibility tree $G'$ exists, the algorithm $\atwo$ for Phase II is well-defined and correct.
\end{lemma}
\begin{proof}
    As the first step, we show that each time $\atwo$ is supposed to choose some output label that has to be contained in $\comp^e(u)$ for some half-edge $(u,e)$, we have $\comp^e(u) \neq \undec$.
    From the description of $\atwo$, it follows that when the above situation occurs, the half-edge $(u,e)$ in question is 1) incident to the root, or 2) incident to some node of the form $\pred_p$ for some pointer $p = (v,w) \in \ps_{\fin}$ and $e$ does not lie on the path from $v$ to $w$.
    If $(u, e)$ is incident to the root, i.e., if $u = \rooot$, then \Cref{lem:attheroot}, together with the design of the update rules and the fact that at the end of Phase I no active pointer remains, implies that $\comp^e(u) \neq \undec$.
    In the other case, observe that the existence of $p$ implies that at some point during Phase I, $\pred_p$ was a $2$-node with relevant in-edge on the path from $v$ to $w$ (due to the design of the update rules).
    Now, \Cref{obs:undec} yields the desired inequality.
    
    Next, we show by induction that at each step of $\atwo$, labels with the required properties are available and the partial solution produced by choosing such labels is part of some correct global solution.
    Note that, due to \Cref{lem:between,obs:summary}, it suffices to show the induction step for each set $\succc(p)$ of pointers separately as the processing of two distinct sets $\succc(p), \succc(p')$ is independent of each other (due to the facts that the half-edges considered when processing $\succc(p)$ are separated by some already selected output label from the half-edges considered when processing $\succc(p')$, and that the correctness of a solution for a compatibility tree is defined via constraints on edges and constraints on nodes).
    For simplicity, we will use the notation from the description of $\atwo$ in the following.
    
    For the base of the induction, observe that for the first step of point (a), there is a choice of output labels with the described properties due to \Cref{lem:encodecomplete} (in conjunction with the very first step of this proof) and the fact that a correct solution for $G'$ exists.
    Moreover, the obtained partial solution is part of some correct global solution due to the properties required in the first step of point (a) (and \Cref{lem:encodecomplete}).
    Observe also that the same holds for the second step of point (a): \Cref{lem:encodecomplete} together with the fact that there exists a correct global solution that respects the partial coloring computed so far ensures that labels with the required properties exist; the properties (and \Cref{lem:encodecomplete}) in turn imply that the new obtained partial solution is still part of some correct global solution.
    (A bit more concretely, the fact that the partial solution after the first step of point (a) is extendable to a correct global solution implies that there must be a label $\ell^*$ as described in the second step of point (a) since $\pairs_{p_j}$ contains all label pairs that can be completed to a correct solution inside the subtree hanging from $\last_{p_j}$ (by \Cref{lem:encodecomplete}) and the condition $(\ell^*) \in S_{u_j}$ just states that the output is correct ``at $u_j$''; the fact that the resulting output label pair (at the half-edges $(r,\first_{p_j})$ and $(\last_{p_j},u_j)$) is contained in $\pairs_{p_j}$ implies that the new partial solution can still be extended to a correct global solution, again due to the characterization of $\pairs_{p_j}$ given in \Cref{lem:encodecomplete}.)
    
    For the induction step, an analogous argumentation shows that, also for point (b), the extendability of the obtained partial solution implies the availability of labels with the stated properties, and the properties of the labels imply the extendability of the new obtained partial solution to a correct global solution.
    
    To prove the correctness of the algorithm and conclude the proof, given the above, it suffices to show that each half-edge becomes labeled at some point.
    To this end, observe that \Cref{lem:between,obs:summary}, and the definition of $\timey(\cdot)$ imply that any pointer $p \in \ps_{\fin}$ is contained in at most one $\timey(i)$.
    Since there are only finitely many pointers in $\ps_{\fin}$, there must be some positive $i$ such that $\timey(i) = \emptyset$; hence $\atwo$ terminates.
    Since \Cref{lem:between} implies that all basic pointers have been processed when $\atwo$ terminates, we obtain the desired statement that each half-edge becomes labeled.
\end{proof}

\paragraph*{Bounding the Number of Iterations.}

What is left to be done is to bound the number of iterations in Phase II.

\begin{lemma}\label{lem:runtimephasetwo}
    Algorithm $\atwo$ terminates after $O(\log n)$ iterations.
\end{lemma}
\begin{proof}
    Let $p = (u,v)$ be some pointer contained in some $\timey(i)$ satisfying $\pred_p \neq \bot$, and let $p' \in \succc(p)$.
    Let $j$, resp.\ $j'$, denote the first iteration in Phase I such that $p$, resp.\ $p'$, was an active pointer at the end of iteration $j$ (possibly $j = 0$ or $j' = 0$).
    Our first goal is to show that $j' < j$.
    
    To this end, observe that, due to the design of the update rules in Phase I (and the fact that the (active) outgoing pointer of a node can only grow), $p$ must be the result of the merge operation $\merge(q, q')$, where $q = (u, \pred_p)$ and $q' = (\pred_p, v)$, and this merge must have been performed in iteration $j$.
    Hence, $q$ and $q'$ must have been active pointers in iteration $j - 1$, which implies that, if $p' \in \{ q, q'\}$, then $j' < j$, as desired.
    Thus, assume that $p' \notin \{ q, q'\}$, which, by \Cref{obs:summary}, implies that $p' = (w, \pred_p)$ for some leaf $w$, and $\last_{p'}$ does not lie on the path from $u$ to $v$.
    Since, at the end of iteration $j - 1$, node $\pred_p$ is a $2$-node with its relevant in-edge lying on the path from $u$ to $v$ (as $\merge(q, q')$ is performed in iteration $j$), each active pointer $q''$ ending in $\pred_p$ at the end of iteration $j - 1$ must satisfy that $\last_{q''}$ lies on the path from $u$ to $v$, and this fact cannot change in the further course of Phase I.
    Hence, $p'$ must have been active before iteration $j - 1$, and, again, we obtain $j' < j$.
    
    By \Cref{obs:summary}, we can conclude that for any pointer $p'$ in any $\timey(i')$, there must be a pointer $p$ in $\timey(i'-1)$ such that the first iteration in Phase I at the end of which $p$ was active is strictly larger than the first iteration in Phase I at the end of which $p'$ was active (where we consider the starting configuration to be ``at the end of iteration $0$'').
    This implies that $\timey(i'' + 2) = \emptyset$, where $i''$ is the number of iterations in Phase I.
    Hence, Algorithm $\atwo$ terminates after $O(\log n)$ iterations, by \Cref{lem:runtimephaseone}.
\end{proof}

\subsection{Implementation in the \mpc Model} \label{sec:highimple}
In this section, we describe how to implement algorithm $\fA$ from \Cref{sec:highalgo} in the low-space \mpc model. As the implementation of the first phase of $\fA$ (rooting the input tree) is provided in \Cref{sec:rooting}, we can focus on the remaining two phases.

We first consider the part of $\fA$ from \Cref{sec:trafothere}, i.e., the part of the algorithm where we bring the number of nodes down to $O(n/\log n)$, and the part where the solution computed on the compatibility tree in \Cref{sec:highmainalgo} is transformed into a solution for the considered \lcl.
Both parts can be implemented in a straightforward manner, due to \Cref{lem:gi}, unless we run into memory issues due to the fact that we have to execute $O(\log \log n)$ of the iterations described in \Cref{sec:trafothere}, instead of just one.
Note that a global memory overhead can only be possibly produced by the new edges that are introduced in the graphs $G_1, G_2, \dots, G_t$ since the node set only shrinks during that process and the memory required for the ``compatibility information'' of the compatibility graphs cannot be asymptotically larger than the memory required for storing the edges.
Moreover, as the total number of edges (produced during the process) that are incident to any particular vertex is at most $O(\log \log n)$, we cannot run into issues with the local memory.
Hence, it suffices to show that the total number of edges produced during the $O(\log \log n)$ iterations does not exceed $\Theta(n)$ (or $\Theta(m)$, since $n$ and $m$ are asymptotically equal for trees).
However, this directly follows from the fact that for each new edge that is introduced during those iterations, a node is removed.

For the remainder of the section, we consider the part of $\fA$ from \Cref{sec:highmainalgo}, i.e., algorithm $\fA'$ that solves the compatibility tree.
We start by collecting all information that has to be stored during $\fA'$.
For simplicity, we already assign this information to the nodes of the compatibility tree $G'$.
We observe that unless a node has to store more than $n^{\delta}$ words or the total amount of information to be stored is in $\omega(n)$, the algorithm can be naively implemented by standard techniques.
Soon we will see that the total amount of information to be stored is in $O(n)$, and the only issue to be taken care of is that the local memory of ``nodes'' is exceeded.
We will explain later how to resolve this issue.

For Phase I (see \Cref{sec:phaseone}), we maintain two pieces of information, namely
\begin{enumerate}
    \item a set of pointers, and
    \item the sets $\comp(u)$.
\end{enumerate}
For each pointer $p$, some additional information ($\pairs_p$, $\pred_p$, $\first_p$, $\last_p)$, and whether it is active or not) has to be stored, but since the memory required for this additional information is only a constant multiple of the memory required to store the pointer itself (in particular as there are only a constant number of output labels in $\sout$), we can ignore this information.
Moreover, as the design of the update rules in Phase I ensures, the number of pointers produced in each iteration is upper bounded by the number of nodes of the compatibility tree, which is $O(n/\log n)$, by \Cref{cor:fewernodes}.
Since, by \Cref{lem:runtimephaseone}, there are only $O(\log n)$ iterations in Phase I, the total number of pointers that have to be stored is in $O(n)$; hence, our global memory of $O(m)$ is not exceeded.
We already note that we will store each pointer $p = (u,v)$ at both of its endpoints; the overhead introduced by this does not change the required global memory asymptotically.

Together with a set $\comp(u)$, we also have to store the information about all the leaf-root pointers in $\ps_{\fin}$ that end in $u$ (in order to perform the steps in Phase II); however, by the design of the update rules of Phase I and the fact that degrees are bounded (and $|\sout|$ is constant), all of this information requires just a constant number of words to be stored.
We will store each set $\comp(u)$ and its associated information in node $u$; as the required amount of memory per node is constant (in words), we can ignore this information in the remainder of this section.

For the implementation of Phase II, the information stored in Phase I is still required, but will not be changed or expanded.
Note that the characterization of the pointers in $\timey(i)$, given by \Cref{obs:summary}, provides a straightforward implementation: the pointers that are processed in iteration $1$ are easily identified (as they are the only leaf-root pointers in $\ps_{\fin}$), and the pointers to be processed in any later iteration are precisely those in a set of the form $\succc(p)$, where $p$ is a pointer processed in the previous iteration.
Observe also that, by \Cref{lem:nicesplit,obs:summary} and the fact that degrees are bounded, each node is involved in the processing of only a constant number of pointers in each iteration.
Hence, each iteration of Phase II can be easily implemented in a constant number of rounds.

From the description of the update rules of Phase I, it is easy to see that, again, each iteration (now of Phase I) can be performed in constant time provided that we can perform all merge operations (i.e., step~\ref{step2}) in constant time.
From the above discussion, it follows that, in order to obtain the desired runtime of $O(\log n)$ rounds for the complete algorithm $\fA$, the only thing left to be done is to show that we can perform the merge operations in each iteration in constant time while storing the pointers in a way that does not exceed the local memory of the machines.
In the following, we explain how to achieve this.
Note that the merge operation that creates a pointer $(u,w)$ from $(u,v)$ and $(v,w)$ can be understood as $v$ forwarding (the head of) pointer $(u,v)$ to node $w$.

\paragraph*{Pointer Forwarding Tree.}
Our approach relies on a broadcast tree structure that we create for each node with a large number of incoming active pointers (note that this structure is not the same as the \mpc broadcast tree of \Cref{sec:broadcasttree} that contains all machines in the network).
Each node $v$ creates a $n^\delta$-ary virtual rooted tree, where the idea is to store the incoming active pointers in the leaves of the tree.
Importantly, different nodes might be stored at different machines, but since the number of incoming active pointers is bounded by $O(n)$, the communication tree has constant depth which allows us to perform operations efficiently.

To perform the actual pointer forwarding, consider a non-virtual node $v$ and suppose that $v$ wants to forward its incoming active pointers to the non-virtual node $u$ in which the active pointer starting from $v$ ends.
Notice that the active pointers incoming to $v$ are stored in the leaves of the virtual tree $T_v$ rooted at $v$ and similarly for $u$ in the tree $T_u$ rooted at $u$.
Now, we can simply attach the tree $T_v$ to the node in $T_u$ that currently stores the pointer $(u,v)$ (for $u$).
Thereby, the pointers previously incoming to $v$ are now stored in the broadcast tree of $u$, and are therefore incoming to $u$.

This might, however, result in the depth of the broadcast tree increasing by an additive term of $1/\delta$.
To mend this, consider the following balancing process.
\begin{observation}
    Let $T_u$ be a virtual rooted tree of depth $d = O(1/\delta)$ of at most $n$ nodes.
    Then, in $O(1/\delta)$ rounds, we can reduce the depth to $2/\delta$ such that all the leaf nodes of $T_u$ are still leaf nodes.
\end{observation}
\begin{proof}
    The root initiates the following operation.
    First, using converge-cast, it learns the number $n_v$ of nodes in each (virtual) tree rooted from each of its child $v$.
    Then, the root creates a new $n^\delta$-ary virtual tree $T^*$ of depth $O(1/\delta)$.
    Proportionally to the number $n_v$, root $u$ assigns subtrees of $T^*$ to child $v$, such that all incoming pointers corresponding to $v$ fit into the subtree.
    Clearly, this is possible since $n_v \leq n$ and the new virtual node is assigned to at most $n^\delta$ children of $u$.
    This process is recursively continued until the leaves of $T^*$ are assigned to the leaves of $T_u$.
    Then, we can change the pointers from the old broadcast tree nodes to the new ones, and we have obtained our broadcast tree of depth $O(1/\delta)$.
    Notice that a naive implementation results in a $1/\delta$ number of converge-casts, but the number of leaves per subtree can be pre-computed and stored.
\end{proof}

\section{Rake-and-Compress Decomposition} \label{sec:rc}

In this section, we give an $O(\log \log n)$ time deterministic, low-space \mpc rake-and-compress algorithm for constant-degree trees, using $O(m)$ words of global memory. In the rake-and-compress decomposition, the node set $V$ of a graph is split into sets $V_1,V_2,\dots,V_L$ such that $V=V_1 \cup V_2 \cup \dots \cup V_L$ and $V_i \cap V_j = \emptyset$ for all $i \not= j$. Parameter $L$ is called the size of the decomposition and the sets are referred to as layers. This decomposition is an essential tool in our mid-complexity regime. Our algorithm is based on simulating the \local algorithm of Chang and Pettie~\cite{chang}, while paying special attention to memory use.  First, in \Cref{subsec:decomplocal}, we summarize their \local algorithm. Then, in \Cref{subsec:decompmpc}, we describe our low-space \mpc algorithm and prove the following lemma.
 
\begin{lemma} \label{lemma:partitioning}
    There is a deterministic $O(\log \log n)$ time rake-and-compress algorithm for constant-degree trees in the low-space \mpc model using $O(m)$ words of global memory.
\end{lemma}


\subsection{Decomposition in \local} \label{subsec:decomplocal}

The algorithm consists of two steps: a decomposition step, where nodes are partitioned into layers and a post-processing step, where we compute an $(\alpha, \beta)$-independent set (\Cref{def:alpha-beta}), in $O(\log^*N)=O(\log^*n)$ time~\cite{local} and adjust the layers slightly. Recall that $N=\text{poly}(n)$ denotes the size of our ID space.

\begin{definition}[$(\alpha, \beta)$-independent set] \label{def:alpha-beta}
Let $P$ be a path. A set $I \subset V(P)$ is called an $(\alpha, \beta)$-independent set if the following conditions are met: (i) $I$ is an independent set, and $I$ does not contain either endpoint of $P$, and (ii) each connected component induced by $V(P)-I$ has at least $\alpha$ vertices and at most $\beta$ vertices, unless $|V(P)| < \alpha$, in which case $I=\emptyset$.
\end{definition}

\begin{enumerate}
    \item Suppose $l$ is some constant depending on the \lcl problem. The algorithm begins with $U = V(G)$ and $i=1$, repeats Steps (a)--(c) until $U=\emptyset$, then proceeds to Step 2.

    \begin{enumerate}
        \item For each $v \in U$:
        
        \begin{enumerate}
        \item \textsf{Compress}. If $v$ belongs to a path $P$ such that $|V(P)| \geq l$ and $\deg_{U}(u) = 2$ for each $u \in V(P)$, then tag $v$ with $i_C$.
        \item \textsf{Rake}. If $\deg_U(v)=0$, then tag $v$ with $i_R$. If $\deg_U(v)=1$ and the unique neighbor $u$ of $v$ in $U$ satisfies either (i) $\deg_U > 1$ or (ii) $\deg_U = 1$ and $\text{ID}(v) > \text{ID}(u)$, then tag $v$ with $i_R$.
        \end{enumerate}
        \item Remove from $U$ all vertices tagged $i_C$ or $i_R$ and set $i \larr i+1$.
    \end{enumerate}

    \item Initialize $V_i$ as the set of all vertices tagged $i_C$ or $i_R$. 
    The graph induced by $V_i$ consists of unbounded length paths, but we prefer constant length paths. For each edge $\{u,v\}$ such that $v$ is tagged $i_R$ and $u$ is tagged $i_C$, promote $v$ from $V_i$ to $V_{i+1}$. For each path $P$ that is a connected component induced by vertices tagged $i_C$, compute an $(l,2l)$-independent set $I_P$ of $P$, and then promote every vertex in $I_P$ from $V_i$ to $V_{i+1}$.

\end{enumerate}

\begin{observation}\label{obs:decompProperties}
    The following properties of the decomposition are either evident or proven by Chang and Pettie~\cite[Section 3.9]{chang}.
    \begin{itemize}
    \item Define $G_i$ as the graph induced by nodes in layer $i$ or higher: $\bigcup_{j=i}^L V_j$. For each $v \in V_i$, $\deg_{G_i}(v) \leq 2$.
    \item Define $\mathcal{P}_i$ as the set of connected components (paths) induced by the nodes in $V_i$ with more than one node. For each $P \in \mathcal{P}_i$, $l \leq |V(P)| \leq 2l$ and $\deg_{G_i}(v)=2$ for each node $v \in V(P)$.
    \item The graph $G_L$ contains only isolated nodes, i.e., $\mathcal{P}_L=\emptyset$.
    \item At least a constant $\Omega(1/l)$ fraction of vertices in $U$ are eliminated in each iteration, resulting in a runtime of $O(\log n)$ and decomposition size $L=O(\log n)$.
\end{itemize}

As a consequence, each vertex $v \in V_i$ falls into exactly one of two cases: (i) $v$ has $\deg_{G_i}(v) \leq 1$ and has no neighbor in $V_i$, or (ii) $v$ has $\deg_{G_i}(v) = 2$ and is in some path $P \in \mathcal{P}_i$.

\end{observation}

\subsection{Decomposition in \mpc} \label{subsec:decompmpc}

Function \textsf{Peel($r$)}: compute the lowest $r$ layers of the decomposition by simulating $r/l$ steps (more precisely, only Step 1) of the \local algorithm $l$ times. The $1/l$ factor stems from the fact that in each round, all nodes have to know if they belong to a path of length $\geq l$. Recall that by \Cref{obs:decompProperties}, at least a constant $\Omega(1/l)$ fraction of nodes are eliminated in each simulation. When taking a closer look into~\cite{chang}, the exact fraction is $1/2(l+1) \geq 1/4l$. Hence, we can state that at most a constant $1-1/4l$ fraction of nodes is left in the graph after each step of the \local algorithm. Set constant $c \larr \text{argmin}_c \{c : (1-1/4l)^c < 1/\Delta \}$, and observe that since $\Delta$ and $l$ are constants, $c$ is also constant. The algorithm performs the following and then simulates Step 2 of the \local algorithm.

\begin{enumerate}
    \item For $i = 1,\dots,\log \log n^\delta$: in phase $i$, all nodes apply \textsf{Peel($2^{i-1}$)} $2c$ times, and then perform one graph exponentiation step.
    \item All nodes apply \textsf{Peel($\delta \log n$)} until the graph is empty.
\end{enumerate}

\paragraph*{Correctness and Runtime.} 
Correctness follows from~\cite[Section 3.9]{chang}, as we only simulate their algorithm.
Let us bound the time complexity. In Step 1, during any phase $i$, each node sees its $2^{i-1}$-radius neighborhood due to graph exponentiation. This vision enables each node to apply \textsf{Peel($2^{i-1}$)} $2c$ times, which takes constant time. After $\log \log n^\delta$ phases, all nodes see their $\delta \log n$-radius neighborhoods and Step 1 terminates. In Step 2, nodes apply \textsf{Peel($\delta \log n$)} until the graph is empty, which takes $O(1/\delta)=O(1)$ time, since there are $O(\log n)$ layers in the decomposition in total. Since the vision of each node is $\delta \log n$ and Step 2 of the \local algorithm takes $O(\log^*N) = O(\log^*n)$ time, we can simulate it in $O(1)$ time. We conclude that the algorithm runs in $O(\log \log n)$ time.

\paragraph*{Memory.} In Step 1, during any phase $i$, applying \textsf{Peel($2^{i-1}$)} $2c$ times results in simulating Step 1 of the \local algorithm $2^{i}c$ times. Hence, after applying \textsf{Peel($2^{i-1}$)}, in any phase $i$, there are at most $n \cdot (1-1/4l)^{2^i c} < n / \Delta^{2^i}$ nodes left in the graph. Since the graph exponentiation step of phase $i$ requires at most $\Delta^{2^{i}}$ memory per node, we conclude that each phase (and hence the whole algorithm), requires at most $O(m)$ words of global memory. After $\log \log n^\delta$ phases, all nodes see their $\delta \log n$-radius neighborhoods. Since $\Delta$ is constant, the $\delta \log n$-radius neighborhood of any node contains at most $O(n^\delta)$ nodes and the local memory is never violated.

\section{The Mid-Complexity Regime}\label{sec:mid}

In this section, we will prove that all \lcl problems on trees with deterministic complexity $n^{o(1)}$ in the \local model can be solved deterministically in $O(\log \log n)$ time in the low-space \mpc model using $O(m)$ words of global memory. In previous work, Chang and Pettie showed that in the \local model, there are no \lcl problems on trees whose complexity lies between $\omega(\log n)$ and $n^{o(1)}$~\cite{chang}.
They obtain their result by showing that any problem in this range admits a canonical way to solve it using the decomposition described in \Cref{sec:rc} and a careful method that labels the tree, layer by layer (of the decomposition).

\begin{theorem} \label{thm:mid}
	Consider an \lcl problem $\Pi$ in constant-degree trees with deterministic complexity $n^{o(1)}$ in the \local model. There is a deterministic low-space \mpc algorithm that solves $\Pi$ in $O(\log \log n)$ time using $O(m)$ words of global memory.
\end{theorem}

In order to prove the theorem above, we only need to show that the canonical algorithm of Chang and Pettie~\cite{chang} can be performed deterministically in \mpc using our strict memory parameters. The basis for their algorithm is a rake-and-compress decomposition, which we have already shown can be obtained deterministically in $O(\log \log n)$ time in \mpc with the desired memory parameters (\Cref{sec:rc}). What is left, is to show that the rest of their algorithm can also be performed exponentially faster in the low-space \mpc model. For this, we provide an explicit algorithm.

Let us adopt the same node-labeled \lcl problem definition as Chang and Pettie~\cite{chang}. 
Note that this definition includes port-numberings and is hence equivalent to our definition of half-edge labeled \lcls (\Cref{def:lcl}).

\begin{definition}[\lcl, Chang and Pettie~\cite{chang}] \label{def:newlcl}
Fix a class $\fG$ of possible input graphs and let $\Delta$ be the maximum degree in any such graph. An \lcl problem $\Pi$ for $\fG$ has a radius $r$, constant size input and output alphabets $\sinn$, $\sout$, and a set $\fC$ of acceptable configurations. Note that $\sinn$ and $\sout$ can include $\bot$. Each $C \in \fC$ is a graph centered at a specific vertex, in which each vertex has a degree, a port numbering, and two labels from $\sinn$ and $\sout$. Given the input graph $G(V,E,\phinn)$ where $\phinn: V(G)  \xrightarrow{} \sinn$, a feasible labeling output is any function $\phout :V(G)  \xrightarrow{} \sout$ such that for each $v \in V(G)$, the subgraph induced by $N^{r}(v)$ (denoting the $r$-neighborhood of $v$ together with information stored there: vertex degrees, port numberings, input labels, and output labels) is isomorphic to a member of $\fC$. A complete labeling output is such that for each $v$, $\phout(v) \neq \bot$. An \lcl can be described explicitly by enumerating a finite number of acceptable configurations.
\end{definition}

Let us revisit some other definitions and results of Chang and Pettie~\cite{chang} before introducing our algorithm.

\begin{definition}[Class, Chang and Pettie~\cite{chang}] \label{def:class}
    Consider a rooted tree $T$ with a root $v$ and an \lcl problem $\Pi$ as in \Cref{def:newlcl}. The (equivalence) class of $T$, denoted $\class(T)$ is the set of all possible node labelings of the $r$-hop neighborhood of $v$ such that the labeling can be extended to a complete feasible labeling of $T$ (with respect to the \lcl problem $\Pi$). Note that for constant-degree trees, the number of equivalence classes is constant.
    
\end{definition}

\begin{lemma}\label{lemma: replace}
    Consider a graph $G$, an \lcl problem $\Pi$ as in \Cref{def:newlcl}, and the rake-and-compress decomposition described in \Cref{sec:rc}. Recall that the decomposition is parameterized by a constant $\ell$ that depends on the input $\lcl$ and that $G_i$ denotes the graph induced by the nodes in layer $i$ or higher.
    
    Let $P = (v_1, v_2, \ldots, v_x)$, for some $x \in [\ell, 2\ell]$ be a path induced by nodes in $G_i$ with degree $2$ and let $s = v_1$ and $t = v_x$. Moreover, let graph $T_1 \cup T_2 \cup \ldots \cup T_x$, denoted by $H = (T_1, T_2, \ldots, T_x)$, correspond to the sequence of disjoint trees rooted from nodes $(v_1,v_2, \dots, v_x)$. Then, there exists a tree $H^+$ with two dedicated nodes $s^+$ and $t^+$ such that the following holds.
    
    Let $G_s$ be the connected component of $G - H$ that is adjacent to node $s$ in $G$. Notice that $s \not\in G_s$. Then, if graph $G_s$ is not empty, graph $G_{s^+}$ is created by connecting a copy of $G_s$ to a copy of $H^+$ via a single edge $\{u,s^+\}$ such that $u \in G_{s}$ and $s^+ \in H^+$. The graph $G_t$ is constructed identically but using a disjoint copy of $H^+$.
    
    
    Then, if graphs $G_{s^+}$ and $G_{t^+}$ admit feasible node labelings $\phout^{s^+}$ and $\phout^{t^+}$, then the input graph $G$ admits a feasible node labeling $\phout^*$ such that for any $v \in G - H$, we have $\phi^*(v) = \phout^{s^+} (v)$ if $v \in G_s$ and $\phout^*(v) = \phout^{t^+} (v)$ if $v \in G_t$. Furthermore, the graph $H^+$ can be computed with the knowledge of $\class(T_i)$ for each $i$.
\end{lemma}

\begin{proof}
    The existence of $H^+$ is rigorously proven by Chang and Pettie~\cite[Lemma 9]{chang} through \textit{pumping}, \textit{tree surgery} and \textit{pre-commitment} techniques. The same previous work describes the computation of graph $H^+$~\cite[Proof of Lemma 13]{chang}.
\end{proof}

Note that in \Cref{def:class,lemma: replace} we talk about rooted trees and rooted subtrees hanging from nodes in a path. One could rightfully assume that we either root the tree beforehand or assume a rooted tree as input. However, we do none of the previous. Instead, when talking about a tree $T_v$ rooted at node $v$ in layer $i$, we simply refer to the subgraph induced by nodes in layers $< i$ that is connected to $v$ via one or possibly multiple edges.

\subsection{The Algorithm} \label{subsec:midregime}

Suppose that we are given the rake-and-compress decomposition described in \Cref{sec:rc}. First, we divide the nodes in \emph{batches} according to which layer they belong to in the decomposition. Let $c$ be such that for all $i$, $|\bigcup_{j \geq i + c} V_j| \leq |\bigcup_{j \geq i} V_j| \cdot \Delta^{-1}$, i.e., if we remove $c$ layers from the decomposition, the number of nodes drops by a factor of at least $\Delta$. By \Cref{obs:decompProperties}, we know that $c$ is a constant. Let us define nodes in layers $V_1,\dots,V_c$ as batch $B_0$. For $i>0$ and as long as $\Delta^{2^i} \leq n^{\delta}$, let us define nodes in layers $V_{(2^i-1) \cdot c + 1}, \ldots, V_{(2^{i+1}-1) \cdot c}$ as batch $i$; see \Cref{fig:mid}. Note that their are $O(\log \log n)$ batches as  defined previously, and assuming that there is enough layers, batch $i$ always consists of $2^i \cdot c$ layers. All nodes that do not belong to any batch, as defined previously, are defined as batch $B_L$. The algorithm starts with running $O(\log \log n)$ phases, each of which is executed in a constant number of \mpc rounds and consists of the following steps.

\begin{figure}
    \centering
    \includegraphics[width=0.8\textwidth]{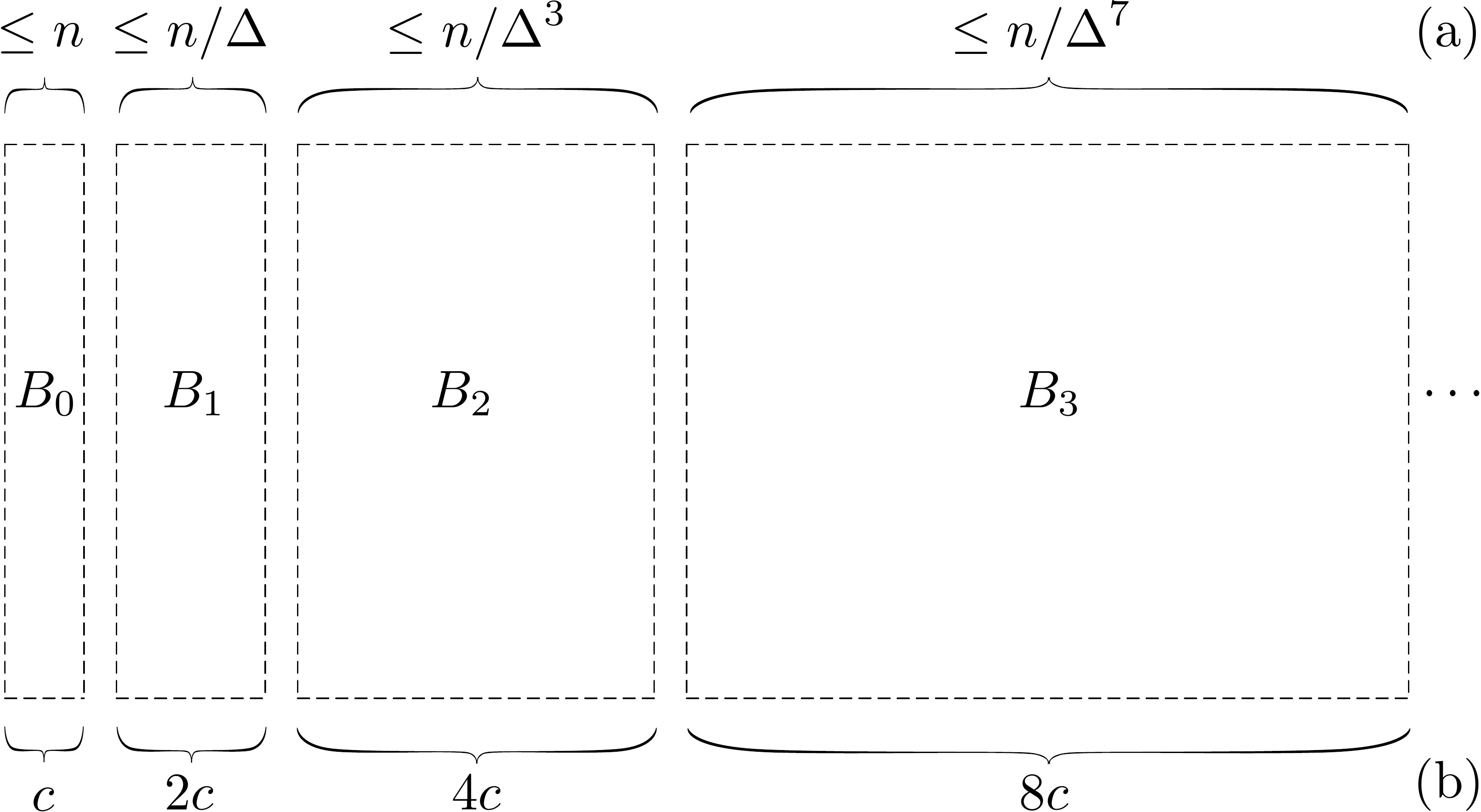}
    \caption{The first four batches, where each batch $i$ contains $2^i \cdot c$ layers and at most $n/\Delta^{2^i-1}$ nodes: (a) number of nodes in a batch, (b) number of layers in a batch.}
    \label{fig:mid}
\end{figure}

\begin{enumerate}
    \item In the start of phase $i$, our communication graph is $G^{2^i}$. We process the $i$:th batch by simulating $2^i \cdot c$ iterations of the following local process. If a node $v$ has neighbors only in higher layers, it locally computes $\class(v)$ and informs its (unique) neighbor about the class. If a node $v$ has neighbors in the same layer, then it must be part of a constant length path. Each node that is not an endpoint of such a path locally computes its class and informs the endpoints of the class. Once an endpoint $s$ learns of all classes on the path, it locally replaces the path with graph $H^+$ as described in \Cref{lemma: replace}. Then $s$ locally computes $\class(s)$ and informs its parent (if any) of the class.
    \item As the second step, for all but the last phase, each node in batch $B_{j > i}$ and in batch $B_L$ performs one step of graph exponentiation. Note that the nodes that have computed their classes, i.e., nodes in  batch $B_{j \leq i}$ do not participate in graph exponentiation. Thus, we obtain $G^{2^{i+1}}$ as our communication graph for batches $>i$ and batch $B_L$. We make the exception for the last phase as we do not want to violate local memory of the nodes in batch $B_L$.
\end{enumerate}

Since we process one batch in each phase, all numbered batches are processed in $O(\log \log n)$ phases. If batch $B_L$ is non-empty, its communication graph is $G^{\Theta(\log n)}$ and the algorithm proceeds by simulating the local process described above until all nodes in $B_L$ has derived their class.

Once all nodes have computed their class, we process the batches in the reversed order. The (local) roots begin by choosing a label that can be extended to a valid labeling on the whole graph. Then, once a node learns the feasible label of its parent (nodes have at most one neighbor in a higher layer), it can choose an extendable label. Similarly, once the parents of both endpoints of a layer-induced path of $H^+$ have decided on their labels, the endpoints choose a valid labeling of the original layer-induced path of $H$.

\paragraph*{Correctness.} By the definition of a class and by \Cref{lemma: replace}, the operations in Step 1 are always possible and all nodes can compute their class. The existence of the valid labels for each node is again provided by the definition of a class and by \Cref{lemma: replace}.

\paragraph*{Runtime.} Let us bound the time complexity of the algorithm. First, let us analyze the complexity of computing the class of each node. Each phase $i$ indeed takes constant time, since simulating $2^i \cdot c$ iterations of the local process takes $O(1)$ time due to the communication graph in batches $\geq i$ being $G^{2^i}$. As mentioned previously, after $O(\log \log n)$ phases, all batches except $B_L$ have been processed and the communication graph of batch $B_L$ is $G^{\Theta(\log n)}$. Note that by \Cref{obs:decompProperties}, batch $B_L$ contains $O(\log n)$ layers. Hence, simulating the local process of our algorithm on our communication graph $G^{\Theta(\log n)}$ takes constant time, after which all nodes have computed their class. 

Next, let us analyze the complexity of computing the label of each node. Using the communication graph created throughout the algorithm, we process the batches in the reverse order, which requires the same number of \mpc rounds as computing the class. 

Since obtaining the rake-and-compress decomposition required for this process takes $O(\log \log n)$ time, the overall runtime of the algorithm is $O(\log \log n)$.

\paragraph*{Memory.} Next, let us analyze the memory requirement. The local memory bound is respected, since a batch $i$ is defined such that $\Delta^{2^i} \leq n^\delta$ and a node in batch $i$ executes at most $i$ steps of graph exponentiation, resulting in a neighborhood containing at most $\Delta^{2^i}$ nodes. Note that batch $B_L$ performs the same number of graph exponentiation steps as does the last numbered batch and hence, local memory bounds are not violated. Now for the global memory bound. Consider batch $i$. By the definition of $c$, we have that batch $i$ contains at most $|\bigcup_{j \geq (2^i-1) \cdot c} V_j| \leq n \cdot \Delta^{-(2^i-1)}$ nodes. Since batch $i$ executes at most $i$ steps of graph exponentiation, the memory required to store the communication graph per numbered batch is at most $|\bigcup_{j \geq (2^i-1) \cdot c} V_j| \cdot \Delta^{2^i} \leq n \cdot \Delta$. Now recall that nodes in batch $B_L$ performs the same number of graph exponentiation steps as does the last numbered batch. Observe that the size of $B_L$ has the same upper bound as the last numbered batch. Hence, storing the communication graph of $B_L$ also requires at most $n \cdot \Delta$ memory. The number of batches is clearly upper bounded by the number of layers and hence, the global memory use is bounded by $O(n \log n)$.

By separating the nodes in the first $c \log \log n$ layers and computing their classes prior to the execution and their labels after the other nodes have been handled, we can drop the requirement to sharp $O(n)$, and hence $O(m)$. This is evident as separating the first $c \log \log n$ layers leaves us with at most $n/\log n$ nodes. This scheme contributes only an additive $O(\log \log n)$ term to the runtime. 

Another source of memory issues during execution could be sending messages of $\omega(1)$ size. Fortunately, our messages only contain class information (\Cref{def:class}), and since there are only a constant number of classes, our messages are of constant size.

%


\clearpage
\phantomsection
\bibliographystyle{alphaurl}
\bibliography{mpc-tree-hierarchy}

\newcommand{\etalchar}[1]{$^{#1}$}
\begin{thebibliography}{BCHM{\etalchar{+}}21}

\bibitem[ABI86]{Alon86}
Noga Alon, L\'{a}szl\'{o} Babai, and Alon Itai.
\newblock {A Fast and Simple Randomized Parallel Algorithm for the Maximal
  Independent Set Problem}.
\newblock {\em J. Algorithms}, 1986.
\newblock \href {https://doi.org/10.1016/0196-6774(86)90019-2}
  {\path{doi:10.1016/0196-6774(86)90019-2}}.

\bibitem[ANOY14]{mpcrefine1}
Alexandr Andoni, Aleksandar Nikolov, Krzysztof Onak, and Grigory Yaroslavtsev.
\newblock Parallel algorithms for geometric graph problems.
\newblock In {\em STOC}, 2014.
\newblock \href {https://doi.org/10.1145/2591796.2591805}
  {\path{doi:10.1145/2591796.2591805}}.

\bibitem[Aut21]{low-complex}
Anonymous Authors.
\newblock Personal communication, work in submission, 2021.

\bibitem[BBD{\etalchar{+}}19]{Behnezhad19}
Soheil Behnezhad, Sebastian Brandt, Mahsa Derakhshan, Manuela Fischer,
  MohammadTaghi Hajiaghayi, Richard~M. Karp, and Jara Uitto.
\newblock {Massively Parallel Computation of Matching and MIS in Sparse
  Graphs}.
\newblock In {\em PODC}, 2019.
\newblock \href {https://doi.org/10.1145/3293611.3331609}
  {\path{doi:10.1145/3293611.3331609}}.

\bibitem[BBE{\etalchar{+}}20]{balliu2020}
Alkida Balliu, Sebastian Brandt, Yuval Efron, Juho Hirvonen, Yannic Maus,
  Dennis Olivetti, and Jukka Suomela.
\newblock {Classification of Distributed Binary Labeling Problems}.
\newblock In {\em DISC}, 2020.
\newblock \href {https://doi.org/10.1145/3382734.3405703}
  {\path{doi:10.1145/3382734.3405703}}.

\bibitem[BBO{\etalchar{+}}21]{Balliu21}
Alkida Balliu, Sebastian Brandt, Dennis Olivetti, Jan Studeny, Jukka Suomela,
  and Aleksandr Tereshchenko.
\newblock {Locally Checkable Problems in Rooted Trees}.
\newblock In {\em PODC}, 2021.
\newblock \href {https://doi.org/10.1145/3465084.3467934}
  {\path{doi:10.1145/3465084.3467934}}.

\bibitem[BBOS18]{tree3}
Alkida Balliu, Sebastian Brandt, Dennis Olivetti, and J.~Suomela.
\newblock {Almost Global Problems in the LOCAL Model}.
\newblock In {\em DISC}, 2018.
\newblock \href {https://doi.org/10.4230/LIPIcs.DISC.2018.9}
  {\path{doi:10.4230/LIPIcs.DISC.2018.9}}.

\bibitem[BCHM{\etalchar{+}}21]{lclcongest}
Alkida Balliu, Keren Censor-Hillel, Yannic Maus, Dennis Olivetti, and Jukka
  Suomela.
\newblock {Locally Checkable Labelings with Small Messages}.
\newblock In {\em DISC}, 2021.
\newblock \href {https://doi.org/10.4230/LIPIcs.DISC.2021.8}
  {\path{doi:10.4230/LIPIcs.DISC.2021.8}}.

\bibitem[BFH{\etalchar{+}}16]{sinkless16}
Sebastian Brandt, Orr Fischer, Juho Hirvonen, Barbara Keller, Tuomo
  Lempi{\"{a}}inen, Joel Rybicki, Jukka Suomela, and Jara Uitto.
\newblock {A Lower Bound for the Distributed Lov{\'{a}}sz Local Lemma}.
\newblock In {\em STOC}, 2016.
\newblock \href {https://doi.org/10.1145/2897518.2897570}
  {\path{doi:10.1145/2897518.2897570}}.

\bibitem[BFU19]{sirocco}
Sebastian Brandt, Manuela Fischer, and Jara Uitto.
\newblock {Breaking the Linear-Memory Barrier in MPC: Fast MIS on Trees with
  Strongly Sublinear Memory}.
\newblock In {\em SIROCCO}, 2019.
\newblock \href {https://doi.org/10.1007/978-3-030-24922-9_9}
  {\path{doi:10.1007/978-3-030-24922-9_9}}.

\bibitem[BHK{\etalchar{+}}17]{Brandt17}
Sebastian Brandt, Juho Hirvonen, Janne~H. Korhonen, Tuomo Lempi\"{a}inen,
  Patric~R.J. \"{O}sterg\r{a}rd, Christopher Purcell, Joel Rybicki, Jukka
  Suomela, and Przemys\l{}aw Uzna\'{n}ski.
\newblock {LCL Problems on Grids}.
\newblock In {\em PODC}, 2017.
\newblock \href {https://doi.org/10.1145/3087801.3087833}
  {\path{doi:10.1145/3087801.3087833}}.

\bibitem[BHK{\etalchar{+}}18]{BHKLOS18}
Alkida Balliu, Juho Hirvonen, Janne~H. Korhonen, Tuomo Lempi{\"{a}}inen, Dennis
  Olivetti, and Jukka Suomela.
\newblock {New Classes of Distributed Time Complexity}.
\newblock In {\em STOC}, 2018.
\newblock \href {https://doi.org/10.1145/3188745.3188860}
  {\path{doi:10.1145/3188745.3188860}}.

\bibitem[BHOS19]{tree2}
Alkida Balliu, Juho Hirvonen, Dennis Olivetti, and Jukka Suomela.
\newblock {Hardness of Minimal Symmetry Breaking in Distributed Computing}.
\newblock In {\em PODC}, 2019.
\newblock \href {https://doi.org/10.1145/3293611.3331605}
  {\path{doi:10.1145/3293611.3331605}}.

\bibitem[BKS17]{mpcrefine2}
Paul Beame, Paraschos Koutris, and Dan Suciu.
\newblock Communication steps for parallel query processing.
\newblock {\em J. ACM}, 2017.
\newblock \href {https://doi.org/10.1145/3125644} {\path{doi:10.1145/3125644}}.

\bibitem[CC21]{coy2021deterministic}
Sam Coy and Artur Czumaj.
\newblock {Deterministic Massively Parallel Connectivity}, 2021.
\newblock \href {http://arxiv.org/abs/2108.04102} {\path{arXiv:2108.04102}}.

\bibitem[CDP21a]{componentstable}
Artur Czumaj, Peter Davies, and Merav Parter.
\newblock {Component Stability in Low-Space Massively Parallel Computation}.
\newblock In {\em PODC}, 2021.
\newblock \href {https://doi.org/10.1145/3465084.3467903}
  {\path{doi:10.1145/3465084.3467903}}.

\bibitem[CDP21b]{detcol}
Artur Czumaj, Peter Davies, and Merav Parter.
\newblock {Improved Deterministic {$(\Delta+1)$} Coloring in Low-Space MPC}.
\newblock In {\em PODC}, 2021.
\newblock \href {https://doi.org/10.1145/3465084.3467937}
  {\path{doi:10.1145/3465084.3467937}}.

\bibitem[CFG{\etalchar{+}}19]{Chang2019}
Yi-Jun Chang, Manuela Fischer, Mohsen Ghaffari, Jara Uitto, and Yufan Zheng.
\newblock {The Complexity of {$(\Delta+1)$} Coloring in Congested Clique,
  Massively Parallel Computation, and Centralized Local Computation}.
\newblock In {\em PODC}, 2019.
\newblock \href {https://doi.org/10.1145/3293611.3331607}
  {\path{doi:10.1145/3293611.3331607}}.

\bibitem[Cha20]{Chang2020}
Yi-Jun Chang.
\newblock {The Complexity Landscape of Distributed Locally Checkable Problems
  on Trees}.
\newblock In {\em DISC}, 2020.
\newblock \href {https://doi.org/10.4230/LIPIcs.DISC.2020.18}
  {\path{doi:10.4230/LIPIcs.DISC.2020.18}}.

\bibitem[CHL{\etalchar{+}}19]{tree1}
Yi-Jun Chang, Qizheng He, Wenzheng Li, Seth Pettie, and Jara Uitto.
\newblock {Distributed Edge Coloring and a Special Case of the Constructive
  Lovász Local Lemma}.
\newblock {\em ACM Trans. Algorithms}, 2019.
\newblock \href {https://doi.org/10.1145/3365004} {\path{doi:10.1145/3365004}}.

\bibitem[CKP19]{CKP19}
Yi{-}Jun Chang, Tsvi Kopelowitz, and Seth Pettie.
\newblock {An Exponential Separation between Randomized and Deterministic
  Complexity in the {LOCAL} Model}.
\newblock {\em {SIAM} J. Comput.}, 2019.
\newblock \href {https://doi.org/10.1137/17M1117537}
  {\path{doi:10.1137/17M1117537}}.

\bibitem[CP17]{chang}
Yi-Jun Chang and Seth Pettie.
\newblock {A Time Hierarchy Theorem for the LOCAL Model}.
\newblock In {\em FOCS}, 2017.
\newblock \href {https://doi.org/10.1109/FOCS.2017.23}
  {\path{doi:10.1109/FOCS.2017.23}}.

\bibitem[DFKL21]{spanner}
Michal Dory, Orr Fischer, Seri Khoury, and Dean Leitersdorf.
\newblock {Constant-Round Spanners and Shortest Paths in Congested Clique and
  MPC}.
\newblock In {\em PODC}, 2021.
\newblock \href {https://doi.org/10.1145/3465084.3467928}
  {\path{doi:10.1145/3465084.3467928}}.

\bibitem[DG04]{dg04}
Jeffrey Dean and Sanjay Ghemawat.
\newblock {MapReduce: Simplified Data Processing on Large Clusters}.
\newblock In {\em OSDI}, 2004.
\newblock \href {https://doi.org/10.1145/1327452.1327492}
  {\path{doi:10.1145/1327452.1327492}}.

\bibitem[FG17]{FischerG17}
Manuela Fischer and Mohsen Ghaffari.
\newblock {Sublogarithmic Distributed Algorithms for Lov{\'{a}}sz Local Lemma,
  and the Complexity Hierarchy}.
\newblock In {\em DISC}, 2017.
\newblock \href {https://doi.org/10.4230/LIPIcs.DISC.2017.18}
  {\path{doi:10.4230/LIPIcs.DISC.2017.18}}.

\bibitem[GGJ20]{GGC20}
Mohsen Ghaffari, Christoph Grunau, and Ce~Jin.
\newblock {Improved {MPC} Algorithms for {MIS}, Matching, and Coloring on Trees
  and Beyond}.
\newblock In {\em DISC}, 2020.
\newblock \href {https://doi.org/10.4230/LIPIcs.DISC.2020.34}
  {\path{doi:10.4230/LIPIcs.DISC.2020.34}}.

\bibitem[GGR21]{GhaffariGR21}
Mohsen Ghaffari, Christoph Grunau, and V{\'{a}}clav Rozhon.
\newblock {Improved Deterministic Network Decomposition}.
\newblock In {\em SODA}, 2021.
\newblock \href {https://doi.org/10.1137/1.9781611976465.173}
  {\path{doi:10.1137/1.9781611976465.173}}.

\bibitem[Gha16]{Ghaffari16}
Mohsen Ghaffari.
\newblock {An Improved Distributed Algorithm for Maximal Independent Set}.
\newblock In {\em SODA}, 2016.
\newblock \href {https://doi.org/10.1137/1.9781611974331.ch20}
  {\path{doi:10.1137/1.9781611974331.ch20}}.

\bibitem[GHK18]{GHK18}
Mohsen Ghaffari, David~G. Harris, and Fabian Kuhn.
\newblock {On Derandomizing Local Distributed Algorithms}.
\newblock In {\em FOCS}, 2018.
\newblock \href {https://doi.org/10.1109/FOCS.2018.00069}
  {\path{doi:10.1109/FOCS.2018.00069}}.

\bibitem[GKU19]{focs}
Mohsen Ghaffari, Fabian Kuhn, and Jara Uitto.
\newblock {Conditional Hardness Results for Massively Parallel Computation from
  Distributed Lower Bounds}.
\newblock In {\em FOCS}, 2019.
\newblock \href {https://doi.org/10.1109/FOCS.2019.00097}
  {\path{doi:10.1109/FOCS.2019.00097}}.

\bibitem[GRB22]{lclcomplete}
Christoph Grunau, Vaclav Rozhon, and Sebastian Brandt.
\newblock {The Landscape of Distributed Complexities on Trees and Beyond},
  2022.
\newblock \href {http://arxiv.org/abs/2202.04724} {\path{arXiv:2202.04724}}.

\bibitem[GS17]{GhaSu17}
Mohsen Ghaffari and Hsin{-}Hao Su.
\newblock {Distributed Degree Splitting, Edge Coloring, and Orientations}.
\newblock In {\em SODA}, 2017.
\newblock \href {https://doi.org/10.1137/1.9781611974782.166}
  {\path{doi:10.1137/1.9781611974782.166}}.

\bibitem[GSZ11]{broadcast}
Michael~T. Goodrich, Nodari Sitchinava, and Qin Zhang.
\newblock {Sorting, Searching, and Simulation in the Mapreduce Framework}.
\newblock In {\em ICAAC}, 2011.
\newblock \href {https://doi.org/10.1007/978-3-642-25591-5_39}
  {\path{doi:10.1007/978-3-642-25591-5_39}}.

\bibitem[GU19]{GU19}
Mohsen Ghaffari and Jara Uitto.
\newblock {Sparsifying Distributed Algorithms with Ramifications in Massively
  Parallel Computation and Centralized Local Computation}.
\newblock In {\em SODA}, 2019.
\newblock \href {https://doi.org/10.1137/1.9781611975482.99}
  {\path{doi:10.1137/1.9781611975482.99}}.

\bibitem[IBY{\etalchar{+}}07]{Isard:2007}
Michael Isard, Mihai Budiu, Yuan Yu, Andrew Birrell, and Dennis Fetterly.
\newblock {Dryad: Distributed Data-Parallel Programs from Sequential Building
  Blocks}.
\newblock In {\em EuroSys}, 2007.
\newblock \href {https://doi.org/10.1145/1272996.1273005}
  {\path{doi:10.1145/1272996.1273005}}.

\bibitem[KLM{\etalchar{+}}14]{pathexp}
Raimondas Kiveris, Silvio Lattanzi, Vahab Mirrokni, Vibhor Rastogi, and Sergei
  Vassilvitskii.
\newblock {Connected Components in {MapReduce} and Beyond}.
\newblock In {\em SOCC}, 2014.
\newblock \href {https://doi.org/10.1145/2670979.2670997}
  {\path{doi:10.1145/2670979.2670997}}.

\bibitem[KSV10]{KarloffSV10}
Howard~J. Karloff, Siddharth Suri, and Sergei Vassilvitskii.
\newblock {A Model of Computation for MapReduce}.
\newblock In {\em SODA}, 2010.
\newblock \href {https://doi.org/10.1137/1.9781611973075.76}
  {\path{doi:10.1137/1.9781611973075.76}}.

\bibitem[Lin87]{linial}
Nathan Linial.
\newblock {Distributive Graph Algorithms -- Global Solutions from Local Data}.
\newblock In {\em FOCS}, 1987.
\newblock \href {https://doi.org/10.1109/SFCS.1987.20}
  {\path{doi:10.1109/SFCS.1987.20}}.

\bibitem[Lin92]{local}
Nathan Linial.
\newblock {Locality in Distributed Graph Algorithms}.
\newblock {\em {SIAM} J. Computing}, 1992.
\newblock \href {https://doi.org/10.1137/0221015} {\path{doi:10.1137/0221015}}.

\bibitem[Lub85]{Luby85}
Michael Luby.
\newblock {A Simple Parallel Algorithm for the Maximal Independent Set
  Problem}.
\newblock In {\em STOC}, 1985.
\newblock \href {https://doi.org/10.1145/22145.22146}
  {\path{doi:10.1145/22145.22146}}.

\bibitem[LW10]{wattenhofer}
Christoph Lenzen and Roger Wattenhofer.
\newblock {Brief Announcement: Exponential Speed-Up of Local Algorithms Using
  Non-Local Communication}.
\newblock In {\em PODC}, 2010.
\newblock \href {https://doi.org/10.1145/1835698.1835772}
  {\path{doi:10.1145/1835698.1835772}}.

\bibitem[NS95]{NaorS95}
Moni Naor and Larry~J. Stockmeyer.
\newblock {What Can be Computed Locally?}
\newblock {\em {SIAM} J. Computing}, 1995.
\newblock \href {https://doi.org/10.1137/S0097539793254571}
  {\path{doi:10.1137/S0097539793254571}}.

\bibitem[Pel00]{peleg}
David Peleg.
\newblock {\em {Distributed Computing: A Locality-Sensitive Approach}}.
\newblock Society for Industrial and Applied Mathematics, 2000.
\newblock \href {https://doi.org/10.1137/1.9780898719772}
  {\path{doi:10.1137/1.9780898719772}}.

\bibitem[RG20]{RozhonG20}
V{\'{a}}clav Rozhon and Mohsen Ghaffari.
\newblock {Polylogarithmic-Time Deterministic Network Decomposition and
  Distributed Derandomization}.
\newblock In {\em STOC}, 2020.
\newblock \href {https://doi.org/10.1145/3357713.3384298}
  {\path{doi:10.1145/3357713.3384298}}.

\bibitem[RVW18]{Roughgarden18}
Tim Roughgarden, Sergei Vassilvitskii, and Joshua~R. Wang.
\newblock {Shuffles and Circuits (On Lower Bounds for Modern Parallel
  Computation)}.
\newblock {\em J. ACM}, 2018.
\newblock \href {https://doi.org/10.1145/3232536} {\path{doi:10.1145/3232536}}.

\bibitem[Whi12]{White:2012}
Tom White.
\newblock {\em {Hadoop: The Definitive Guide}}.
\newblock O'Reilly Media, Inc., 2012.
\newblock \url{https://dl.acm.org/doi/book/10.5555/1717298}.

\bibitem[ZCF{\etalchar{+}}10]{ZahariaCFSS10}
Matei Zaharia, Mosharaf Chowdhury, Michael~J. Franklin, Scott Shenker, and Ion
  Stoica.
\newblock {Spark: Cluster Computing with Working Sets}.
\newblock In {\em HotCloud}, 2010.
\newblock \url{https://dl.acm.org/doi/10.5555/1863103.1863113}.

\end{thebibliography}

\newpage
\appendix

\section{The Broadcast Tree} \label{sec:broadcasttree}

A commonly used subroutine in the \mpc model is the broadcast (or converge-cast) tree. The \mpc broadcast tree is an $n^\delta$-ary communication tree with depth $1/\delta=\Theta(1)$. It enables broadcasting constant sized messages to all machines in constant time while respecting $O(n^\delta)$ local memory and $O(m)$ global memory bounds. It is often assumed to exist without much discussion~\cite{GGC20, sirocco, Behnezhad19, broadcast}. Let us reason how a broadcast tree can be constructed. In the \mpc model, we allow global communication between all machines, and in order for it to be feasible, all machines have to somehow know the address of every other machine. Locally, each machine can do the following computation.

\begin{enumerate}
	\item Sort all addresses by some deterministic criteria and map them to set $\{1,2,\dots,M\}$.
	\item Enumerate the $n^\delta$-ary broadcast tree in a breadth-first fashion, after which each machine knows the position of all machines in the broadcast tree, in particular, its own.
\end{enumerate}

Hence, a broadcast tree can always be used as a subroutine, even in the low-space \mpc model.

\end{document}